\documentclass[
superscriptaddress,
reprint,
amsmath,amssymb,
aps,
prx, floatfix, showkeys
]{revtex4-2}
\usepackage{caption}
\usepackage{graphicx}
\usepackage{dcolumn}
\usepackage{bm}
\usepackage{xcolor}
\usepackage{amsthm}
\usepackage{mathtools}
\usepackage{mathrsfs}
\usepackage{booktabs}
\usepackage{physics}  
\usepackage{bbm}      
\usepackage{multirow}
\usepackage{array}
\usepackage{tikz}
\usetikzlibrary{positioning} 
\usepackage{graphicx}
\usepackage{subcaption}
\usepackage{float} 

\newtheorem{theorem}{Theorem}

\newtheorem{corollary}[theorem]{Corollary}

\newtheorem{observation}[theorem]{Observation}

\newcommand{\eat}[1]{}

\begin{document}

\title{Localization Without Disorder: Quantum Walks on Structured Graphs}

\author{Shyam Dhamapurkar}
\email{shyamd@cmi.ac.in}
\affiliation{Department of Computer Science, Chennai Mathematical Institute, Chennai, India.}

\author{K. Venkata Subrahmanyam }
\email{kv@cmi.ac.in}
\affiliation{Department of Computer Science, Chennai Mathematical Institute, Chennai, India.}

\date{\today}

\begin{abstract}
Continuous-time quantum walks (CTQWs) exhibit localization phenomena that differ fundamentally from their classical counterparts, yet the precise relationship between network structure, spectral degeneracy, and confined dynamics remains incompletely understood. In this work, we present a complete analytical characterization of localization in CTQWs on two highly symmetric graph families: barbell graphs and star-of-cliques graphs. These networks combine pronounced spectral degeneracy with modular structure, enabling exact diagonalization and explicit computation of both eigenstate and dynamical inverse participation ratios (IPRs). Our analysis reveals that localization is governed by the interplay between degenerate subspaces, which generate families of confined modes, and hybridization between invariant subspaces, which redistributes spectral weight. Notably, the dynamical IPR can exceed expectations based solely on eigenstate IPRs, demonstrating that coherent superposition within degenerate eigenspaces enhances confinement. By connecting IPR values to the effective number of vertices visited, we provide a structural diagnostic for predicting quantum transport outcomes in modular networks, establishing that connectivity alone can determine where and how strongly a quantum walk localizes.
\end{abstract}


\maketitle

\section{Introduction}
\label{sec:introduction}

\begin{figure}[!htbp]
\centering
\begin{subfigure}{0.32\textwidth}
    \centering
    \includegraphics[trim = 120 280 130 245, clip, width=\linewidth]{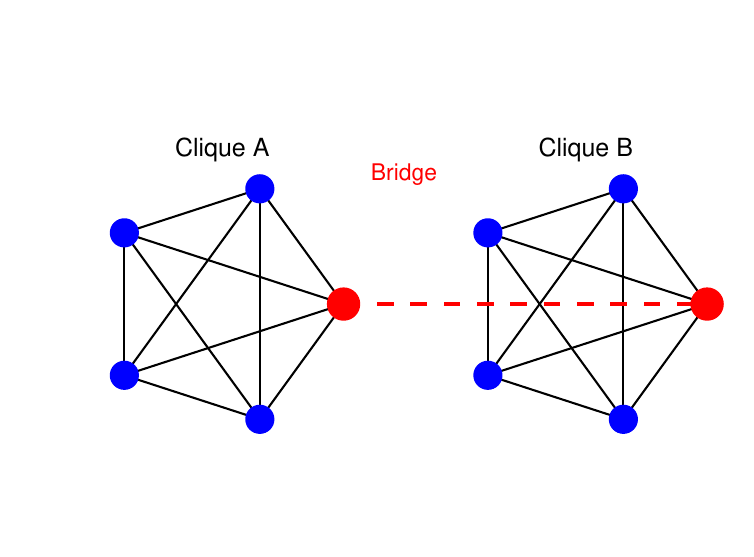}
    \caption{Barbell graph $B(n)$}
    \label{fig:barbell}
\end{subfigure}
\hfill
\begin{subfigure}{0.32\textwidth}
    \centering
    \includegraphics[trim = 120 230 130 190, clip, width=\linewidth]{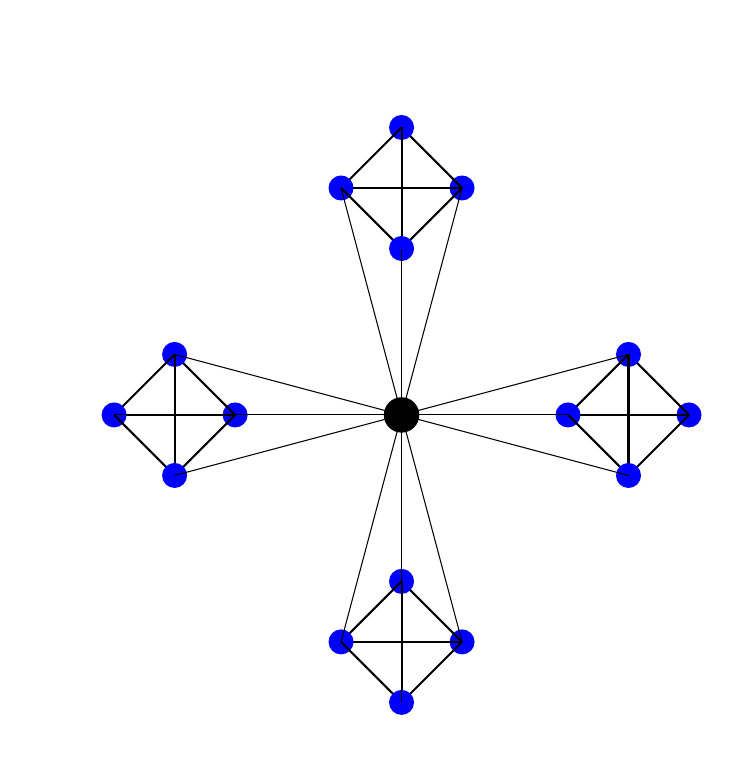}
    \caption{Star-of-cliques: Full connection}
    \label{fig:star_full}
\end{subfigure}
\hfill
\begin{subfigure}{0.32\textwidth}
    \centering
    \includegraphics[trim = 120 240 130 185, clip, width=\linewidth]{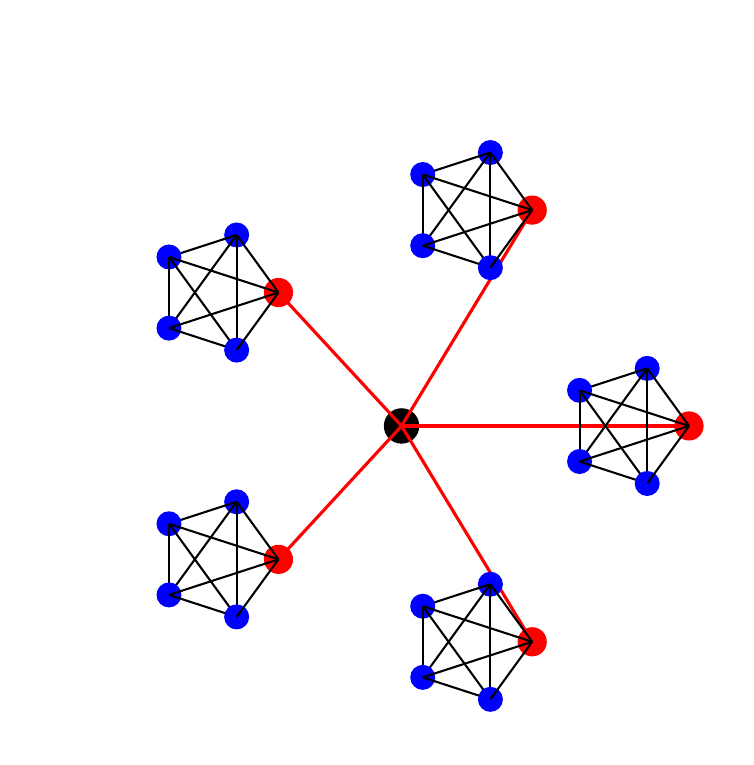}
    \caption{Star-of-cliques: Single connection}
    \label{fig:star_single}
\end{subfigure}
\caption{Graph structures: (a) Barbell graph $B(n)$; (b) Full-connection star-of-cliques; (c) Single-connection star-of-cliques.}
\label{fig:graph_structures}
\end{figure}

Quantum walks provide a framework for understanding coherent transport
on discrete structures and have been studied extensively in both
algorithmic and physical contexts~\cite{RaghavanMotwani,Watrous,Meyer,Nayak,Aharonov,Szegedy}. Unlike classical
random walks, which converge to stationary distributions governed by
Markov chain theory~\cite{RaghavanMotwani}, continuous-time quantum walks (CTQWs)
evolve unitarily and do not converge pointwise in time~\cite{Nayak,childs-notes}.
Instead, their long-time averaged distribution---often referred to as
the limiting distribution---captures the effective outcome of quantum
transport~\cite{Nayak,AAKU,Richter2007}. The structure of this distribution depends
sensitively on spectral properties of the underlying Hamiltonian,
particularly eigenvalue degeneracies and the overlap of the initial
state with the associated eigenspaces.
In continuous-time quantum walks (CTQWs)~\cite{FarhiGutmann}, the dynamics are generated by a Hamiltonian derived from the graph structure, which encodes the connectivity of the underlying network. Although CTQWs do not converge pointwise in time, their long-time averaged probability distribution---the limiting distribution---always exists as shown in~\cite{Aharonov}. This distribution captures the extent to which quantum probability remains localized or spreads over the network and depends sensitively on the spectral structure of the Hamiltonian, particularly eigenvalue degeneracies and the overlap of the initial state with the corresponding eigenspaces.
Localization in quantum walks has traditionally been associated with
disorder, where random potentials or couplings induce interference that
suppresses transport~\cite{Anderson1958,Abrahams1979,Asboth2012,Kendon2007}. Such \emph{disorder-induced localization} has been extensively analyzed using spectral diagnostics and long-time dynamical behavior~\cite{Mulken2011}. More recently, it has been
shown that strong confinement can arise even in the absence of disorder,
purely as a consequence of symmetry and spectral degeneracy~\cite{balachandran2024disorder,Coutinho2022,Keating2007}. This establishes that graph geometry alone can be sufficient to induce robust localization in quantum dynamics, a phenomenon known as \emph{disorder-free localization}.
A standard quantitative measure of localization is the inverse
participation ratio (IPR), defined both for individual eigenstates and
for time-averaged dynamics~\cite{Thouless1974,Wegner1980,Evers2008}. The eigenstate IPR measures the
spatial extent of Hamiltonian eigenvectors and has proven particularly
useful in complex networks where large degenerate eigenspaces support
localized modes associated with structural symmetries~\cite{Yadav2015,Marrec2017,Sarkar2018,BuenoHatano2020}.
Two prominent mechanisms in such networks are community-driven confinement, where densely connected clusters trap probability~\cite{Zhang2021}, and symmetry-protected localization arising from degenerate null-eigenvalue subspaces~\cite{Coutinho2022}. The dynamical IPR, defined from the long-time averaged distribution,
quantifies confinement of a walk initialized at a given vertex and
provides a direct connection between spectral structure and transport
efficiency~\cite{Mulken2011,Zhang2021,BottcherPorter2025,Faccin2013,Keating2007}. While localized eigenstates form the building blocks of localization, the resulting dynamics also depend crucially on interference between eigenvectors and on the degeneracy structure of the spectrum~\cite{Mulken2007,BottcherPorter2025}. These static and dynamical perspectives are unified by the limiting distribution, which represents the effective outcome of quantum transport after temporal averaging and underlies applications in quantum search, state transfer, and quantum memory~\cite{Chandrashekar2015,PR,Venegas_Andraca_2012}.
In this work, we analyze localization in CTQWs on two highly structured
graph families: barbell graphs and star-of-cliques graphs. These
networks combine strong symmetry with pronounced spectral degeneracy,
enabling exact spectral decomposition and explicit evaluation of both
eigenstate and dynamical IPR. Barbell graphs, consisting of two complete graphs connected by a single bridge edge, serve as canonical models for transport through bottlenecks. Star-of-cliques graphs combine clustered subgraphs with a central hub and can be realized with different connectivity patterns, enabling comparison of multiple localization routes within a unified framework. By comparing a barbell graph with two
variants of a star-of-cliques construction, we isolate how small changes
in connectivity alter degeneracy patterns and hybridization between
eigenspaces.
Our analysis reveals a unifying mechanism: localization is governed by
the interplay between degeneracy, symmetry, and interference. Degenerate
subspaces generate families of confined modes, while hybridization
between invariant subspaces redistributes spectral weight and modifies
long-time trapping. The dynamical IPR can exceed naive expectations
based solely on individual eigenstate IPR values, demonstrating that
coherent superposition within degenerate eigenspaces enhances
confinement.
By relating IPR values directly to the effective number of vertices
visited (via \(1/\mathrm{IPR}\)), we provide a structural diagnostic for
predicting quantum transport outcomes in modular networks. The results
show that connectivity alone, even in the absence of disorder, can
determine where and how strongly a quantum walk localizes.

The paper is organized as follows: In Section~\ref{sec:introduction} we review basic definitions of CTQW, time averaged distribution and IPR in quantum walks. In Section 
~\ref{sec:families} we define the network/graph families considered in the paper. In Sections ~\ref{sec:barbell:all}, ~\ref{sec:var1:all}, ~\ref{sec:var2:all} we analyze the eigenstate IPR and dynamical IPR calculations done in the Appendix for the the Barbell graph, and Star of Cliques variant 1 and Star of Cliques variant 2  respectively. Finally in Section ~\ref{sec:conclusion} we summarize our findings and outlines future directions.

\section{Preliminaries}
\label{sec:preliminaries}

We introduce the mathematical definitions and tools used throughout this work.
We do not provide proofs of many standard results, instead pointing the reader to
the original sources.

\subsection{Continuous-Time Quantum Walks}
\label{sec:ctqw}

For a continuous-time quantum walk on a graph $G=(V,E)$ with $|V|=N$, the Hamiltonian $H$ is typically taken to be the normalized adjacency matrix $\tilde{M}$. The continuous-time quantum walk then evolves according to the Schr\"odinger equation
\begin{equation}
\frac{d}{dt}|\psi(t)\rangle = -iH|\psi(t)\rangle,
\end{equation}
with initial state $|\psi(0)\rangle = |j\rangle$ localized at vertex $j$ \cite{FarhiGutmann}. 

Let the spectral decomposition of $H$ be
\begin{equation}
H = \sum_{k=1}^{K} \lambda_k \Pi_k,
\end{equation}
where $\lambda_1 < \lambda_2 < \cdots < \lambda_K$ are the distinct eigenvalues and
\begin{equation}
\Pi_k = \sum_{r=1}^{m_k} \dyad{\phi_k^{(r)}}
\end{equation}
are orthogonal projectors onto eigenspaces of dimension $m_k$, with eigenstates $|\phi_k^{(r)}\rangle$.

The time-evolution operator can be written as 
\begin{equation}
U(t) = e^{-iHt} = \sum_{k=1}^{K} e^{-i\lambda_k t} \Pi_k.
\end{equation}
For an initial state $\ket{\psi_0} = \ket{j}$, the evolved state at time $t$ is
\begin{equation}
\ket{\psi(t)} = U(t)\ket{j}.
\end{equation}

The transition probability from vertex $j$ to vertex $i$ at time $t$ is
\begin{equation}
\pi_{ij}(t) = |\braket{i}{\psi(t)}|^2 = |\langle i|e^{-iHt}|j\rangle|^2.
\end{equation}

Expanding in the eigenbasis $\{|\phi_\mu\rangle\}_{\mu=1}^N$ of $H$ (with eigenvalues $\lambda_\mu$), we have
\begin{equation}
\pi_{ij}(t) = \sum_{\mu=1}^N \sum_{\nu=1}^N \langle i|\phi_\mu\rangle\langle \phi_\mu|j\rangle
\langle j|\phi_\nu\rangle\langle \phi_\nu|i\rangle e^{-i(\lambda_\mu-\lambda_\nu)t}.
\end{equation}

Since the evolution is unitary, $\pi_{ij}(t)$ does not converge pointwise in time. Following standard analyses of quantum walks~\cite{Aharonov, PR}, we consider the time-averaged distribution. The long-time averaged transition probability is

\begin{align}
\label{long_time_average}
\overline{\pi}_{ij} &= \lim_{T\to\infty}\frac{1}{T}\int_0^T \pi_{ij}(t)\,dt \\
&= \sum_{\mu,\nu=1}^N \langle i|\phi_\mu\rangle\langle \phi_\mu|j\rangle
\langle j|\phi_\nu\rangle\langle \phi_\nu|i\rangle \delta_{\lambda_\mu,\lambda_\nu},    
\end{align}

where $\delta_{\lambda_\mu,\lambda_\nu}$ is the Kronecker delta enforcing $\lambda_\mu = \lambda_\nu$. This can be rewritten using the spectral decomposition as
\begin{equation}
\overline{\pi}_{ij} = \sum_{k=1}^{K} \abs{\bra{i}\Pi_k\ket{j}}^2.
\end{equation}

\subsection{Inverse Participation Ratio}
\label{sec:IPR}

The inverse participation ratio (IPR) serves as a fundamental diagnostic for localization in quantum systems. We consider two related notions of IPR: one based on individual eigenstates, and one based on the long-time averaged dynamics.

\subsubsection{Eigenstate IPR}

Let $|\phi_\mu\rangle$ be an eigenvector of $H$ with eigenvalue $\lambda_\mu$. Following Thouless and Wegner~\cite{Thouless1974, Wegner1980}, the scale-invariant eigenstate inverse participation ratio is defined as 
\begin{equation}
\mathrm{IPR}_\mu =
\frac{\sum_{i=1}^N |\langle i|\phi_\mu\rangle|^4}
{\left(\sum_{i=1}^N |\langle i|\phi_\mu\rangle|^2\right)^2}.
\end{equation}
For normalized eigenvectors, this reduces to
\begin{equation}
\boxed{\mathrm{IPR}_\mu = \sum_{i=1}^N |\langle i|\phi_\mu\rangle|^4}.
\end{equation}

The eigenstate IPR ranges from $1/N$, corresponding to complete delocalization where the eigenstate is spread uniformly over all vertices, to $1$, corresponding to localization on a single vertex. As emphasized in \cite{BuenoHatano2020}, the quantity $1/\mathrm{IPR}_\mu$ can be interpreted as the effective number of vertices occupied by the eigenstate.

\subsubsection{Dynamical IPR}

Building on the time-averaged transition probability $\overline{\pi}_{ij}$ from Eq.~\eqref{long_time_average}, we define the \emph{long-time average inverse participation ratio} for an initial vertex $|j\rangle$ as
\begin{equation}
\overline{\mathrm{IPR}}_j = \sum_{i=1}^N (\overline{\pi}_{ij})^2.
\end{equation}

This quantity measures the extent to which a quantum walk initialized at vertex $j$ remains localized in the long-time limit. A large value indicates that the walker's time-averaged probability distribution is concentrated on a few vertices, while a small value indicates delocalization over many vertices.

Substituting the expression for $\overline{\pi}_{ij}$ and expanding the square yields
\begin{multline}
\overline{\mathrm{IPR}}_j = \sum_{i=1}^N \sum_{\mu,\nu=1}^N \sum_{\rho,\sigma=1}^N 
\langle i|\phi_\mu\rangle\langle \phi_\mu|j\rangle
\langle j|\phi_\nu\rangle\langle \phi_\nu|i\rangle \\
\times \langle i|\phi_\rho\rangle\langle \phi_\rho|j\rangle
\langle j|\phi_\sigma\rangle\langle \phi_\sigma|i\rangle
\delta_{\lambda_\mu,\lambda_\nu}\delta_{\lambda_\rho,\lambda_\sigma}.
\end{multline}

Following \cite{Keating2007,BottcherPorter2025,BuenoHatano2020}, this can be organized more explicitly as
\begin{equation}
\begin{aligned}
\overline{\mathrm{IPR}}_j &=
\sum_{i=1}^N \Bigg[ \bigg( \sum_{\mu=1}^N X_\mu^{ij} \bigg)^2
+ 2\bigg( \sum_{\mu=1}^N X_\mu^{ij} \bigg)
\bigg( \sum_{\substack{\mu<\nu \\ \lambda_\mu=\lambda_\nu}} Y_{\mu\nu}^{ij} \bigg) \\
&\quad + \sum_{\substack{\mu<\nu \\ \lambda_\mu=\lambda_\nu}}
\sum_{\substack{r<s \\ \lambda_r=\lambda_s}}
Z_{\mu\nu rs}^{\,ij} \Bigg],
\end{aligned}
\end{equation}
where we define the diagonal contributions
\begin{equation}
X_\mu^{ij} = |\langle i|\phi_\mu\rangle\langle \phi_\mu|j\rangle|^2,
\end{equation}
the degenerate pair terms
\begin{equation}
Y_{\mu\nu}^{ij} = 2\,\Re\!\left( \langle i|\phi_\mu\rangle\langle \phi_\mu|j\rangle
\langle j|\phi_\nu\rangle\langle \phi_\nu|i\rangle \right),
\end{equation}
and the quartet terms
\begin{equation}
Z_{\mu\nu rs}^{\,ij} = 
\begin{cases}
\dfrac{1}{2}\,Y_{\mu\nu}^{ij}Y_{rs}^{\,ij}, & \lambda_\mu-\lambda_\nu = \lambda_r-\lambda_s \\
& \text{or} \ \lambda_\mu-\lambda_\nu = \lambda_s-\lambda_r, \\[8pt]
0, & \text{otherwise}.
\end{cases}
\end{equation}

The first term captures diagonal contributions from single eigenvectors, the second term arises from pairs of degenerate eigenvectors, and the third term involves ``quartets'' of eigenvalues whose differences are equal or opposite. Eigenvalue degeneracies and such quartets are common in highly symmetric networks and can lead to strong localization even when individual eigenstates are extended.

\subsubsection{Relation Between Eigenstate and Dynamical IPR}

We now derive a simple relation between the two notions of IPR introduced above. This will be useful when we compare localization in the families of graphs considered in this paper.

\begin{observation}[Lower Bound]
\label{thm:lb}
Let $|j\rangle = \sum_{\mu=1}^N c_\mu |\phi_\mu\rangle$ with
$c_\mu = \langle \phi_\mu|j\rangle$. Then
\begin{equation}
\overline{\mathrm{IPR}}_j \ge \sum_{\mu=1}^N |c_\mu|^4 \mathrm{IPR}_\mu.
\end{equation}
\end{observation}

\begin{proof}
Using $\overline{\pi}_{ij}
= \sum_\mu |c_\mu|^2 |\langle i|\phi_\mu\rangle|^2$
(neglecting degenerate cross terms for the bound) and applying the
Cauchy--Schwarz inequality,
\begin{equation}
\sum_i (\overline{\pi}_{ij})^2
\ge \sum_i \sum_\mu |c_\mu|^4 |\langle i|\phi_\mu\rangle|^4
= \sum_\mu |c_\mu|^4 \mathrm{IPR}_\mu.
\end{equation}
\end{proof}

\begin{corollary}[Eigenstate Initial Conditions]
If the initial state is an eigenstate $|j\rangle = |\phi_\alpha\rangle$, then
\begin{equation}
\overline{\mathrm{IPR}}_{\phi_\alpha} = \mathrm{IPR}_\alpha.
\end{equation}
\end{corollary}

Thus, eigenstate IPR directly determines dynamical localization for eigenstate initial conditions. More generally, initial states with significant overlap on multiple high-$\mathrm{IPR}_\mu$ eigenvectors yield large dynamical IPR \cite{BottcherPorter2021}. The time-averaged distribution $\overline{\pi}_{ij}$ serves as the fundamental building block for the dynamical IPR, with the relationship
\begin{equation}
\overline{\pi}_{ii} = \mathrm{IPR}_i
\end{equation}
providing a direct link between static eigenstate properties and dynamical measures \cite{Keating2007}. This connection allows us to use eigenstate IPR as a diagnostic for the dynamical behavior of quantum walks initialized at specific vertices.

\section{Networks considered and notation}
\label{sec:families}
We consider two families of graphs, barbell graphs and stars of cliques graphs. We define the two families.
\subsection{Barbell Graphs}
Barbell grapha are a family of graphs, ${\cal B}_n$, $n = 1,\cdots, \cdots$. ${\cal B}_n$, has two cliques $A$, $B$ of size $n$ with an edge joining a vertex in $A$ to a vertex in $B$.
Vertices in $A$ are denoted $\{1, 2, \dots, n\}$ and vertices in $B$ are denoted $\{n+1, n+2, \dots, 2n\}$, with the bridge connecting vertex $n$ to vertex $n+1$. For quantum walk analysis, the Hamiltonian is taken as the normalized adjacency matrix $\tilde{M}$. Understanding the spectral properties of $\tilde{M}$ is crucial for characterizing quantum walk behavior, including mixing times, oscillation frequencies, and transport efficiency.

\subsection{Star of Cliques Graphs}
This $n$-th member of this family of graphs ${\cal SC}_n$. $n = 1 \ldots, \ldots$ has $n$ cliques joined to a central hub vertex. We consider two distinct variants that differ in how the central vertex connects to the $n$ peripheral cliques, each of size $n$. {\bf Variant 1} is the full-connection variant in which the central vertex connects to \emph{all} vertices in each clique. {\bf Variant 2} is the single-connection variant in which the central vertex connects to only \emph{one} vertex in each clique. We use the notation $|0\rangle$ for the central vertex and the notation $|j,k\rangle$ ($k=1,\dots,n$) for the vertices of clique $j$ for $j=1,\dots,n$. The central vertex has degree $d_0 = n^2$ in the full-connection variant and degree $n$ in the single-connection variant, and $d_{(j,k)} = n$. 

\subsection{Spectral and IPR Analysis of the Networks Considered}
In Section~\ref{sec:barbell_spectral_A},~\ref{supp:full_connection}, and~\ref{supp:single_connection} of the appendix, we provide detailed calculations of the spectra of the normalized-adjacency matrices for each family of networks considered. We also do an asymptotic analysis of the spectra and construct an eigen-space basis for each family. For each of the network families considered we give detailed calculations of the eigenstate IPR and the dynamical IPR in Section~\ref{sec:ipr_barbell},~\ref{sec:variant1_ipr}, and ~\ref{ipr_variant2} of the appendix respectively.

\section{Barbell Graph : IPR calculations}
\label{sec:barbell:all}
\subsection{Eigentate IPR}
\label{sec:barbell_ipr_dyn}
From the calculations in Section~\ref{sec:ipr_barbell}, the symmetric eigenvector with eigenvalue $\lambda_{A+} \approx 1$ has approximately uniform support over both cliques and the bridge vertices, and is therefore spatially extended. Its structure $(1,\ldots,1,a,a,1,\ldots,1)$ shows that the two cliques are in phase, as are the bridge vertices. Because there is no phase mismatch across the cut, transport is not suppressed, and the eigenvector remains uniformly distributed across the graph. For this state $\mathrm{IPR}_{\psi_{A+}} \approx \frac{1}{2n}$. For large $n$, this eigenstate is fully delocalized.

For the eigenvectors $|e_{k}^A\rangle$, $k=1,\ldots,n-1$, supported entirely within clique $A$ the inverse participation ratio is
is $\frac{1 + k^3}{k(k+1)^2}$. This quantity is bounded below by $\frac{1}{2}$ for all $k$, and aproaches 1 for large $n$, indicating  spatial localization. These eigenstates therefore remain localized within a single clique, independent of the total graph size. Thus, the degeneracy reflects a family of graph-induced localized modes confined to a single clique.

For the antisymmetric eigenvector $|\psi_{A-}\rangle$,  $\mathrm{IPR}_{\psi_{A-}} \approx \frac{1}{2}$, indicating that this eigenstate is also highly localized, at the two bridge vertices. This is as a result of quantum interference. The eigen vector form $(1,\ldots,1,a,-a,1,\ldots,1)$ indicates that the cliques are in phase, and while the two bridge vertices have the same amplitude there is a phase shift of $\pi$ relative to each other. So quantum transportation is suppressed due to destructive interference and amplitude is trapped.

\subsection{Dynamical IPR}

\textbf{Clique vertex} For a walk initialized at a vertex in either clique, the results of Section~\ref{dynamical_barbell} show that $\overline{\mathrm{IPR}}_{A_1} \approx 0.58$, indicating that, in the long-time limit, the walk remains effectively confined to a constant number of vertices within the clique. This behaviour closely mirrors that of a quantum walk on a complete graph. In that case, the dynamics evolves within a two-dimensional invariant subspace spanned by the starting vertex and the uniform superposition of the remaining vertices. Although the amplitude transiently spreads across the clique, time averaging suppresses these oscillations, and the walk remains predominantly localized at the initial vertex.

In the barbell graph, weak tunneling through the bridge permits slow leakage into the opposite clique, followed by limited spreading there. However, this transport occurs on a much longer timescale, and the time-averaged dynamics remains effectively confined within a constant-sized subset of vertices in the original clique.

\textbf{Bridge Vertex} The bridge vertex has $O(\frac{1}{\sqrt{n}})$ overlap with the symmetric vector, $O(1)$ overlap with the antisymmetric vector and no overlap with other eigenvectors. Since the antisymmetric eigen vector is localized, it follows from observation~\ref{thm:lb} that the $\overline{\mathrm{IPR}}_{Br_A}$ is lower bounded by a constant. In fact $\overline{\mathrm{IPR}_{Br_A}}$ is $\frac{1}{2}$. A walker initialized at $|Br_A\rangle$, has significant overlap with the antisymmetric bridge eigenmode. This mode carries opposite phase on the two bridge vertices, so amplitude attempting to propagate across the cut interferes destructively. The resulting suppression of transport produces a standing-wave pattern localized at the bottleneck. Unlike Anderson localization, which arises from disorder, this confinement is entirely structural and results from coherent interference induced by the graph’s connectivity.
. 

\subsection*{Summary of Localization in the Barbell Graph}

The barbell graph exhibits two distinct and structurally driven localization regimes. 
At the eigenstate level, the global symmetric mode is fully delocalized across both cliques, while the antisymmetric bridge mode forms a standing-wave pattern with opposite phase on the two bridge vertices, suppressing transport across the cut. 
In addition, the highly degenerate clique-supported eigenspaces generate families of modes confined entirely within individual cliques.

These spectral features directly govern the long-time dynamics. 
A walk initialized at a clique vertex remains effectively confined within its original clique, with dynamical IPR bounded away from zero. 
A walk initialized at a bridge vertex strongly overlaps with the antisymmetric standing-wave mode and remains localized at the bottleneck. 
Thus, weak inter-clique connectivity produces persistent dynamical confinement, demonstrating that localization in the barbell graph arises purely from symmetry, degeneracy, and interference rather than disorder.

\section{Star of Cliques, Variant 1 : IPR calculations}
\label{sec:var1:all}
\subsection{Eigenstate IPR}
\label{sec:star1_ipr_dyn}

The eigenvector $|\psi_1\rangle$ has symmetric support on all vertices. Its eigenstate IPR is $\frac{2}{(n+1)^2}$ which vanishes as $n$ grows, indicating complete delocalization across the clique structure. The uniform phase profile implies fully constructive interference throughout the graph: there is no phase mismatch between cliques or vertices, and no suppression of transport. This mirrors the global symmetric mode of the barbell graph, where the in-phase structure across the cut prevents trapping and produces a fully extended eigenstate.  

The eigen vector $|\psi_3^j\rangle$ has support on the the $(j+1)n$ vertices of $j+1$ cliques and lie within a highly degenerate eigenspace. Their inverse participation ratio, $\mathrm{IPR}_{|\psi_3^j\rangle} = \frac{1+j^3}{j(j+1)^2 n}$, scales as $O(1/n)$, indicating that each such eigenstate spreads over $O(n)$ vertices. Thus, these states are delocalized within the cliques they occupy. However, relative to the total number of vertices in the graph, $n^2+1$,
they occupy only a fraction $O(1/n)$, and therefore and therefore exhibit {\bf partial delocalization} at the global scale. This behaviour is analogous to the barbell’s clique-confined modes: constructive interference occurs within individual cliques, while the absence of coherent phase alignment across the entire graph prevents full global extension. 

The antisymmetric eigenvector $\psi_2$ can be written explicitly as
\[
\psi_2
=
\sqrt{\frac{n}{n+1}}
\left(
|0\rangle
-
\frac{1}{\sqrt{n}}\, S
\right),
\qquad
S = \frac{1}{n}\sum_{j=1}^{n}\sum_{k=1}^{n} |j,k\rangle,
\]
where $|0\rangle$ denotes the central vertex and $S$ is the uniform superposition over all clique vertices.

This representation makes the phase structure explicit: the central vertex carries positive amplitude, while the collective clique component appears with opposite phase. The resulting $\pi$ phase difference between the hub and the surrounding cliques induces destructive interference for amplitude attempting to propagate away from the center. Consequently, transport outward from the hub is strongly suppressed, and the eigenstate forms a standing-wave--like pattern localized at the central vertex.

Although $\psi_2$ has formal support on all vertices, its inverse participation ratio is
\[
\mathrm{IPR}_{\psi_2}
=
\frac{n^4+1}{n^2(n+1)^2},
\]
which is bounded below by $\tfrac{1}{2}$ and therefore indicates strong localization. In the limit $n \to \infty$, the weight concentrates almost entirely on $|0\rangle$, and the state becomes effectively localized at the center.

This mechanism parallels the antisymmetric bridge mode in the barbell graph: in both cases, a structural $\pi$ phase mismatch across a bottleneck produces interference-induced confinement without any disorder.

\subsection{Dynamical IPR}
\textbf{Centre Vertex} The localized eigenstate $|\psi_2\rangle$ has asymptotically unit support on the central vertex $|0\rangle$. By Observation~\ref{thm:lb}, this implies that a walk initialized at $|0\rangle$ remains confined to a constant-dimensional subspace in the long-time limit. Indeed, $\overline{\mathrm{IPR}}_{0} \approx 1$, showing that the walker stays localized at the hub. Although the evolution transiently spreads amplitude to the $n^2$ clique vertices, the center overlaps with only two eigenvectors, and time averaging suppresses this spreading. The long-term probability of occupying any given clique is $\frac{1}{2n}$, which vanishes as $n \to \infty$, yielding effective localization at the center.

\textbf{Clique vertices} For any other vertex $|j,k\rangle$, the long-time behaviour mirrors that of the barbell graph. The dominant contributions arise from degenerate eigenmodes confined within clique $j$, and $\overline{\mathrm{IPR}}_{|j,k\rangle} \approx 1$ for large $n$. Thus, a walk starting at a clique vertex remains trapped within that clique and is effectively localized at its starting vertex.

\subsection*{Summary of Localization in Variant 1}

Variant 1 displays a redistribution of localization driven by hybridization between the central vertex and clique subspaces. 
Although large degenerate eigenspaces persist within individual cliques, the coupling to the central vertex alters spectral weights and modifies dynamical confinement.

Eigenstate IPR reveals three regimes: globally extended modes, clique-confined degenerate modes, and bridge-supported antisymmetric modes. 
Dynamically, bridge vertices remain strongly localized due to dominant overlap with antisymmetric standing-wave modes, while clique vertices become asymptotically delocalized as spectral weight spreads across cliques through the central hub.

Thus, unlike the barbell graph, localization in Variant 1 is not uniform across vertex types: hybridization weakens clique confinement while preserving bridge localization.

\section{Eigenstate IPR : Variant 2, Star of Cliques}
\label{sec:var2:all}
\subsection{Eigenstate IPR}
The eigenvectors $|\psi_1\rangle$, $|\psi_2\rangle$, and $|\psi_3\rangle$ have support on the center, the bridge vertices, and all clique vertices. Their amplitudes are uniformly distributed at scale $O(1/n)$ per vertex, so that 
\[
\mathrm{IPR}_{\psi_j} \approx O\!\left(\frac{1}{n^2}\right),
\qquad j=1,2,3,
\]
indicating complete delocalization over the $n^2$ clique vertices. These modes are globally extended and represent fully constructive interference across the graph.

The symmetric eigenvectors $|\phi_4^{(j)}\rangle$ spanning the $(n-1)$–dimensional eigenspace corresponding to $\lambda_4$ have support on the bridge vertices together with the clique vertices. Their inverse participation ratios satisfy 
\[
\mathrm{IPR}_{\phi_4^{(j)}} \approx O\!\left(\frac{1}{n}\right),
\]
showing that each such eigenstate spreads over $O(n)$ vertices. Thus these modes exhibit partial localization: they are delocalized within a collection of cliques but do not extend uniformly over the entire graph. This behaviour parallels the partially localized modes observed in Variant~1.

In contrast, the antisymmetric eigenvectors $|\phi_5^{(j)}\rangle$ (multiplicity $n-1$) are strongly concentrated on the bridge vertices. Their eigenvalue $\lambda_5 \approx -1/n$ is small, corresponding to low-energy modes. The IPR of these states approaches an $O(1)$ constant as $n \to \infty$, demonstrating persistent localization. These form a family of ``bridge-localized'' modes: they are analogous to the antisymmetric bridge eigenvector of the barbell graph, but here the degeneracy produces multiple such standing-wave modes, each confined primarily to a different configuration of bridge vertices. The phase opposition within these vectors induces destructive interference that suppresses transport into the cliques.

Finally, the eigenstates $|\phi_6^{(j,r)}\rangle$ are supported entirely within a single clique and are localized on two vertices. These resemble the clique-confined modes in Variant~1 and arise from the internal symmetry of each complete subgraph.

\subsection{Dynamical IPR}
\textbf{Centre vertex $|0\rangle$.}
The central vertex has significant overlap only with the three globally extended eigenvectors 
$|\psi_1\rangle$, $|\psi_2\rangle$, and $|\psi_3\rangle$, each of which distributes amplitude uniformly at scale $O(1/n)$ over the $n^2$ clique vertices. 
Consequently,
\[
\overline{\mathrm{IPR}}_{0}
=
\Theta\!\left(\frac{1}{n^6}\right),
\]
showing strong dynamical delocalization. 
Although transient spreading occurs across all cliques, the time-averaged distribution becomes nearly uniform over the $n^2$ vertices, and the inverse participation ratio vanishes rapidly as $n \to \infty$. 
This behaviour is driven by constructive interference across the graph and mirrors the fully symmetric delocalized mode.

\medskip

\noindent
\textbf{Bridge vertices $|b_j\rangle$.}
Each bridge vertex has $O(1)$ overlap with the antisymmetric family 
$|\phi_5^{(j)}\rangle$, whose eigenvalue $\lambda_5 \approx -1/n$ is small and whose mass is concentrated primarily on bridge vertices. 
Because these eigenvectors form bridge-localized standing-wave modes with destructive interference suppressing transport into the cliques, they dominate the time-averaged dynamics. 
Accordingly,
\[
\overline{\mathrm{IPR}}_{b_j}
=
1 - \frac{4}{n} + O\!\left(\frac{1}{n^2}\right),
\]
which approaches $1$ as $n \to \infty$, indicating strong localization.

\medskip

\noindent
\textbf{Clique internal vertices $|c_{j,k}\rangle$.}
Vertices internal to a clique overlap predominantly with the highly degenerate clique-confined modes 
$|\phi_6^{(j,r)}\rangle$, which are supported within a single complete subgraph. 
These modes are localized on pairs of vertices and arise from the internal symmetry of each clique. 
As a result,
\[
\overline{\mathrm{IPR}}_{c_{j,k}}
=
1 - \frac{4}{n} + O\!\left(\frac{1}{n^2}\right),
\]
so a walk initialized at such a vertex remains effectively trapped within its clique and is localized at the starting vertex in the long-time limit.

\medskip

\subsection*{Summary of Localization in Variant 2}

Restricting each clique to a single connection to the central vertex fundamentally alters the degeneracy structure. 
The eigenstate analysis reveals globally extended modes, partially delocalized symmetric modes, bridge-localized antisymmetric standing-wave families, and clique-confined degenerate modes.

Dynamically, the center becomes strongly delocalized due to dominant overlap with globally extended eigenvectors, while both bridge and clique vertices exhibit persistent localization governed by degenerate subspaces. 
In contrast to Variant 1, localization is restored at clique vertices, and confinement is redistributed away from the center.

These results show that small structural changes in connectivity reorganize spectral multiplicities and thereby shift localization between vertex classes.

\begin{table*}[ht]
\centering
\caption{Summary of eigenstate and dynamical IPR scaling for barbell graph and star-of-cliques variants. All results are shown in the limit of large graph size $n \to \infty$, where $n$ denotes the size parameter for each graph family.}
\label{tab:ipr_summary}
\begin{tabular}{|l|c|c|c|}
\hline
\textbf{Graph} & \textbf{Vertex/State Type} & \textbf{Eigenstate IPR} & \textbf{Dynamical IPR} \\
\hline
\multirow{3}{*}{Barbell} 
 & Symmetric mode & $O(1/n)$ & — \\
 & Clique-confined modes & $\Omega(1)$ & — \\
 & Antisymmetric bridge mode & $\approx 1/2$ & — \\
\hline
 & Clique vertex & — & $\approx 0.58$ \\
 & Bridge vertex & — & $\approx 1/2$ \\
\hline
\multirow{3}{*}{Star of Cliques (Variant 1)} 
 & Symmetric mode $|\psi_1\rangle$ & $O(1/n^2)$ & — \\
 & Clique-delocalized modes $|\psi_3^j\rangle$ & $O(1/n)$ & — \\
 & Antisymmetric hub mode $|\psi_2\rangle$ & $\Omega(1)$ & — \\
\hline
 & Centre vertex $|0\rangle$ & — & $\approx 1$ \\
 & Clique vertex $|j,k\rangle$ & — & $\approx 1$ \\
\hline
\multirow{4}{*}{Star of Cliques (Variant 2)} 
 & Global modes $|\psi_1\rangle,|\psi_2\rangle,|\psi_3\rangle$ & $O(1/n^2)$ & — \\
 & Partially delocalized modes $|\phi_4^{(j)}\rangle$ & $O(1/n)$ & — \\
 & Bridge-localized modes $|\phi_5^{(j)}\rangle$ & $\Omega(1)$ & — \\
 & Clique-confined modes $|\phi_6^{(j,r)}\rangle$ & $\Omega(1)$ & — \\
\hline
 & Centre vertex $|0\rangle$ & — & $\Theta(1/n^6)$ \\
 & Bridge vertex $|b_j\rangle$ & — & $1 - 4/n + O(1/n^2)$ \\
 & Clique vertex $|c_{j,k}\rangle$ & — & $1 - 4/n + O(1/n^2)$ \\
\hline
\end{tabular}
\end{table*}

\section{Summary and outlook}
\label{sec:conclusion}
We have presented a comparative analysis of eigenstate and dynamical
inverse participation ratios for continuous-time quantum walks on the
barbell graph and two variants of the star-of-cliques graph. The study
demonstrates that modest structural modifications in graph
connectivity---specifically, how cliques are coupled through bridge or
central vertices---produce qualitatively distinct localization regimes
in long-time dynamics.

The quantitative contrasts summarized in Table\textasciitilde I reflect
a consistent spectral mechanism. In the barbell graph, both bridge and
non-bridge vertices exhibit \(O(1)\) dynamical IPR, indicating
persistent confinement arising from weak inter-clique coupling and
antisymmetric standing-wave modes. In Variant\textasciitilde1 of the
star-of-cliques graph, strong hybridization between center and clique
subspaces redistributes spectral weight, yielding localized bridge
vertices but asymptotically delocalized clique vertices. In
Variant\textasciitilde2, restricting each clique to a single connection
to the center alters the degeneracy structure, leading to enhanced
localization at the center while promoting delocalization at bridge
vertices.

The eigenstate IPR clarifies the spectral origin of these behaviors.
Large degenerate subspaces support localized or partially localized
modes whose structure is dictated by symmetry. Crucially, interference
within these degenerate eigenspaces modifies the long-time averaged
distribution: the dynamical IPR can exceed the value suggested by
considering individual eigenstates independently. Localization is
therefore controlled not only by individual eigenvector profiles, but
also by degeneracy multiplicities and hybridization patterns between
invariant subspaces.

A principal conceptual conclusion is that localization is not monotone
in connectivity. Increasing or redistributing inter-clique couplings can
simultaneously enhance confinement at some vertices while inducing
delocalization at others. Thus, degeneracy structure provides a precise
spectral lens through which connectivity governs transport.

From an algorithmic perspective, continuous-time quantum walks have been
studied as primitives in search and Hamiltonian-based computation
\cite{Childs2009,Novo_2015}. Although this work does not establish performance bounds,
the sensitivity of dynamical IPR to degeneracy and coupling patterns
suggests that spectral multiplicity may influence mixing behavior,
hitting times, and state concentration. A systematic comparison between
dynamical IPR, spectral gap, and measured hitting times in highly
degenerate graphs may clarify how structural localization constrains or
enhances algorithmic performance.

Overall, these results reinforce that disorder is not required for
strong localization in quantum transport. Symmetry, degeneracy, and
interference---encoded directly in network architecture---are sufficient
to produce robust and tunable confinement phenomena.

\section{Acknowledgment} Both the authors were supported by a generous grant to CMI by the Infosys Foundation.

\bibliographystyle{apsrev4-2}
\bibliography{reference}

\clearpage

\appendix

\section{ Eigenspectrum for the Barbell Graph $B(n)$}
\label{sec:barbell_spectral_A}

We present the complete spectral analysis for the normalized adjacency matrix $\tilde{M} = \Gamma^{-1/2}M\Gamma^{-1/2}$ of the barbell graph $B(n)$, which consists of two $n$-cliques connected by a single bridge between vertices $n$ (in clique $A$) and $n+1$ (in clique $B$). The matrix $\tilde{M}$ is symmetric and therefore possesses real eigenvalues and orthogonal eigenvectors.

\subsubsection*{Degree Structure}

The degree matrix and its square root are:
\[
\Gamma = \operatorname{diag}\Bigl(\underbrace{n-1,\dots,n-1}_{n-1},\; n,\; n,\; 
\underbrace{n-1,\dots,n-1}_{n-1}\Bigr),
\]
\begin{align*}
\Gamma^{1/2} &= \operatorname{diag}\Bigl(\underbrace{\sqrt{n-1},\dots,\sqrt{n-1}}_{n-1},\; \sqrt{n},\; \sqrt{n},\;\\ 
&\underbrace{\sqrt{n-1},\dots,\sqrt{n-1}}_{n-1}\Bigr).    
\end{align*}

The eigenvalue equation for $\tilde{M}$ is $\tilde{M}\mathbf{y} = \lambda \mathbf{y}$, or equivalently $M\mathbf{y} = \lambda \Gamma \mathbf{y}$.

\subsubsection*{Eigenvalue $\lambda_0 = -\dfrac{1}{n-1}$ (Multiplicity $2n-4$)}

For this eigenvalue, we construct eigenvectors that vanish on the bridge vertices $n$ and $n+1$ and sum to zero on each clique. Explicitly, choose vectors $\mathbf{v} \in \mathbb{R}^{n-1}$ and $\mathbf{w} \in \mathbb{R}^{n-1}$ satisfying $\sum_{i=1}^{n-1} v_i = 0$ and $\sum_{i=n+2}^{2n} w_i = 0$. Then set
\[
\mathbf{y} = (v_1,\dots,v_{n-1},\,0,\,0,\,w_{n+2},\dots,w_{2n})^{\!T}.
\]

For any vertex $i$ in a clique (excluding the bridge vertex), we have $(M\mathbf{y})_i = \sum_{j \in \text{clique}} y_j - y_i = -y_i$ and $(\Gamma\mathbf{y})_i = (n-1)y_i$, which gives $\lambda = -1/(n-1)$. 

We choose the following orthonormal basis of eigenvectors with support in clique $A$:
\begin{align}
|d_k^A\rangle = \frac{1}{\sqrt{k(k+1)}}\left(\sum_{i=1}^k |A_i\rangle - k|A_{k+1}\rangle\right),  k=1,\dots,n-2,
\end{align}
with an identical family $\{|d_k^B\rangle\}_{k=1}^{n-2}$ with support in clique $B$. These vectors satisfy 
\[\langle d_k^A|d_l^A\rangle = \delta_{kl}, \langle d_k^A|d_l^B\rangle = 0\], and have zero overlap with the bridge vertices:
\begin{align}
\langle Br_A|d_k^A\rangle = 0, \langle Br_B|d_k^A\rangle = 0, \langle Br_A|d_k^B\rangle = 0, \langle Br_B|d_k^B\rangle = 0.
\end{align}

This is a degenerate eigenspace of dimension $ 2n-4$.

\subsubsection*{Symmetric Eigenvectors Across the Bridge (Eigenvalues $\lambda_{\pm}^{(1)}$)}

Consider eigenvectors of the form:
\[
\mathbf{y}_{\text{sym}} = \bigl(\underbrace{\sqrt{n-1},\dots,\sqrt{n-1}}_{n-1},\,a\sqrt{n},\,a\sqrt{n},\,
\underbrace{\sqrt{n-1},\dots,\sqrt{n-1}}_{n-1}\bigr)^{\!T}.
\]

For a vertex $i$ in clique $A$ (excluding the bridge), the eigenvalue equation $M\mathbf{y} = \lambda \Gamma \mathbf{y}$ gives:
\[
(M\mathbf{y})_i = (n-2)\sqrt{n-1} + a\sqrt{n}, \qquad (\Gamma\mathbf{y})_i = (n-1)\sqrt{n-1},
\]
leading to:
\begin{equation}
(n-2)\sqrt{n-1} + a\sqrt{n} = \lambda (n-1)\sqrt{n-1}. \tag{S1}
\end{equation}

For the bridge vertex $n$:
\[
(M\mathbf{y})_n = (n-1)\sqrt{n-1} + a\sqrt{n}, \qquad (\Gamma\mathbf{y})_n = n a\sqrt{n},
\]
giving:
\begin{equation}
(n-1)\sqrt{n-1} + a\sqrt{n} = \lambda n a\sqrt{n}. \tag{S2}
\end{equation}

Dividing (S1) by $\sqrt{n-1}$ and (S2) by $\sqrt{n}$ yields the system:
\begin{align}
n-2 + a\sqrt{\frac{n}{n-1}} &= \lambda(n-1), \label{eq:sym1}\\
(n-1)\sqrt{\frac{n-1}{n}} + a &= \lambda n a. \label{eq:sym2}
\end{align}

Solving these equations, we eliminate $a$ to obtain the quadratic:
\begin{equation}
n(n-1)\lambda^{2} - (n^{2} - n - 1)\lambda - 1 = 0,
\end{equation}
whose solutions are:
\begin{equation}
\boxed{\lambda_{\pm}^{(1)} = \frac{n^{2} - n - 1 \pm \sqrt{(n^{2} - n - 1)^{2} + 4n(n-1)}}{2n(n-1)}.}
\end{equation}

The corresponding values of $a$ are given by $a_{\pm} = \left(\lambda_{\pm}^{(1)}(n-1) - (n-2)\right)\sqrt{\frac{n-1}{n}}$.

\subsubsection*{Antisymmetric Eigenvectors Across the Bridge (Eigenvalues $\lambda_{\pm}^{(2)}$)}

Now consider eigenvectors of the form:
\begin{align*}
\mathbf{y}_{\text{anti}} &= \bigl(\underbrace{\sqrt{n-1},\dots,\sqrt{n-1}}_{n-1},\,a\sqrt{n},\,-a\sqrt{n},\,\\
&\underbrace{-\sqrt{n-1},\dots,-\sqrt{n-1}}_{n-1}\bigr)^{\!T}.    
\end{align*}

For a vertex $i$ in clique $A$ (excluding the bridge), the same calculation as in the symmetric case yields:
\begin{equation}
n-2 + a\sqrt{\frac{n}{n-1}} = \lambda(n-1). \tag{A1}
\end{equation}

For the bridge vertex $n$:
\[
(M\mathbf{y})_n = (n-1)\sqrt{n-1} - a\sqrt{n}, \qquad (\Gamma\mathbf{y})_n = n a\sqrt{n},
\]
giving:
\begin{equation}
(n-1)\sqrt{\frac{n-1}{n}} - a = \lambda n a. \tag{A2}
\end{equation}

Solving these equations yields the quadratic:
\begin{equation}
n(n-1)\lambda^{2} - (n^{2} - 3n + 1)\lambda - (2n-3) = 0,
\end{equation}
with roots:
\begin{equation}
\lambda_{\pm}^{(2)} = \frac{n^{2} - 3n + 1 \pm \sqrt{(n^{2} - 3n + 1)^{2} + 4n(n-1)(2n-3)}}{2n(n-1)}.
\end{equation}

The corresponding values of $a$ are 
\[a_{\pm} = \left(\lambda_{\pm}^{(2)}(n-1) - (n-2)\right)\sqrt{\frac{n-1}{n}}.\]

\subsubsection*{Normalized Eigenvectors}

All eigenvectors are already normalized to unit Euclidean norm in the forms given above. We summarize them explicitly:

\text{Degenerate eigenvectors ($\lambda_0 = -1/(n-1)$):}

\begin{align}
|d_k^A\rangle &= \frac{1}{\sqrt{k(k+1)}}\left(\sum_{i=1}^k |A_i\rangle - k|A_{k+1}\rangle\right), \\
|d_k^B\rangle &= \frac{1}{\sqrt{k(k+1)}}\left(\sum_{i=n+2}^{n+k} |B_i\rangle - k|B_{n+k+1}\rangle\right),
\end{align}

 for $k=1,\dots,n-2$, multiplicity $2n-4$.
 
\text{Symmetric eigenvectors ($\lambda_{\pm}^{(1)}$):}
\begin{align}
|s_{\pm}\rangle = &\frac{1}{\sqrt{2(n-1)^2 + 2n a_{\pm}^2}} \,
\bigl(\underbrace{\sqrt{n-1},\dots,\sqrt{n-1}}_{n-1}, \\
&\,a_{\pm}\sqrt{n},\,a_{\pm}\sqrt{n},\underbrace{\sqrt{n-1},\dots,\sqrt{n-1}}_{n-1}\bigr)^{\!T}, \nonumber
\end{align}
with $a_{\pm} = \left(\lambda_{\pm}^{(1)}(n-1) - (n-2)\right)\sqrt{\frac{n-1}{n}}$.

\text{Antisymmetric eigenvectors ($\lambda_{\pm}^{(2)}$):}
\begin{align}
|a_{\pm}\rangle = \frac{1}{\sqrt{2(n-1)^2 + 2n a_{\pm}^2}} \,
\bigl(&\underbrace{\sqrt{n-1},\dots,\sqrt{n-1}}_{n-1},\\
\,a_{\pm}\sqrt{n},\,-a_{\pm}\sqrt{n}, 
&\underbrace{-\sqrt{n-1},\dots,-\sqrt{n-1}}_{n-1}\bigr)^{\!T}, \nonumber
\end{align}
with $a_{\pm} = \left(\lambda_{\pm}^{(2)}(n-1) - (n-2)\right)\sqrt{\frac{n-1}{n}}$.

\subsubsection*{Simplification of the Symmetric Eigenvalues}



Therefore, the symmetric eigenvalues simplify to:
\begin{align}
\lambda_+^{(1)} &= \frac{(n^2 - n - 1) + (n^2 - n + 1)}{2n(n-1)} = \frac{2n^2 - 2n}{2n(n-1)} = 1, \\
\lambda_-^{(1)} &= \frac{(n^2 - n - 1) - (n^2 - n + 1)}{2n(n-1)} = \frac{-2}{2n(n-1)} = -\frac{1}{n(n-1)}.
\end{align}

This is an exact simplification, valid for all $n \ge 2$. The corresponding parameters become:
\begin{align}
a_+ &= \left(1\cdot(n-1) - (n-2)\right)\sqrt{\frac{n-1}{n}} = \sqrt{\frac{n-1}{n}}, \\
a_- &= \left(-\frac{1}{n(n-1)}\cdot(n-1) - (n-2)\right)\sqrt{\frac{n-1}{n}} \\
&= \left(-\frac{1}{n} - n + 2\right)\sqrt{\frac{n-1}{n}}.
\end{align}

\subsubsection*{Asymptotic Forms for Large $n$}

For large $n$, the eigenvalues and eigenvectors simplify to the following leading-order expressions:

\text{Degenerate eigenvectors ($\lambda_0 = -1/(n-1)$):}
\begin{align}
|d_k^A\rangle &= \frac{1}{\sqrt{k(k+1)}}\left(\sum_{i=1}^k |A_i\rangle - k|A_{k+1}\rangle\right),  \\
|d_k^B\rangle &= \frac{1}{\sqrt{k(k+1)}}\left(\sum_{i=n+2}^{n+k} |B_i\rangle - k|B_{n+k+1}\rangle\right), 
\end{align}

for  $k=1,\dots,n-2.$

\text{Symmetric eigenvectors:}
\begin{align}
\lambda_+^{(1)} &= 1, \\
|s_+\rangle &= \frac{1}{\sqrt{2n}}\left(\underbrace{1,\dots,1}_{n-1}, 1, 1, \underbrace{1,\dots,1}_{n-1}\right)^T + O\left(\frac{1}{n^{3/2}}\right), \\[4pt]
\lambda_-^{(1)} &= -\frac{1}{n^2} + O\left(\frac{1}{n^3}\right), \\
|s_-\rangle &= \frac{1}{\sqrt{2}}\left(\underbrace{0,\dots,0}_{n-1}, -1, -1, \underbrace{0,\dots,0}_{n-1}\right)^T + O\left(\frac{1}{n}\right).
\end{align}

\text{Antisymmetric eigenvectors:}
\begin{align}
\lambda_+^{(2)} &= \frac{2}{n} + O\left(\frac{1}{n^2}\right), \\
|a_+\rangle &= \frac{1}{\sqrt{2}}\left(\underbrace{0,\dots,0}_{n-1}, -1, -1, \underbrace{0,\dots,0}_{n-1}\right)^T + O\left(\frac{1}{n}\right), \\[4pt]
\lambda_-^{(2)} &= -\frac{2}{n} + O\left(\frac{1}{n^2}\right), \\
|a_-\rangle &= \frac{1}{\sqrt{2}}\left(\underbrace{0,\dots,0}_{n-1}, -1, 1, \underbrace{0,\dots,0}_{n-1}\right)^T + O\left(\frac{1}{n}\right).
\end{align}

The parameters $a_\pm$ for the symmetric and antisymmetric families satisfy:
\begin{align}
a_+ &\approx 1, \quad a_- \approx -n \quad \text{(for symmetric family)},\\
a_\pm &\approx -n \quad \text{(for antisymmetric family)}.
\end{align}

Note that $|s_-\rangle$ and $|a_+\rangle$ have the same leading-order form but differ in their $O(1/n)$ corrections, which ensure orthogonality. The factor $1/\sqrt{2n}$ in $|s_+\rangle$ arises from normalization over $2n$ vertices, while the factor $1/\sqrt{2}$ in the bridge-localized states arises from normalization over the two bridge vertices.

\section{Eigenspectrum of Variant 1}

\label{supp:full_connection}

We consider a star-of-cliques graph with $n$ cliques, each of size $n$, where the central vertex $0$ connects to \emph{all} vertices in each clique. The total number of vertices is $1 + n^2$. The degrees are:
\begin{align*}
d_0 &= n^2 \quad \text{(center connected to all $n^2$ clique vertices)}, \\
d_{(j,k)} &= n \quad \text{($n-1$ within clique + 1 connection to center)},
\end{align*}
for $j=1,\dots,n$ and $k=1,\dots,n$.

The normalized adjacency matrix $\tilde{M} = \Gamma^{-1/2}M \Gamma^{-1/2}$ has entries:
\[
\tilde{M}_{0,(j,k)} = \frac{1}{\sqrt{d_0 d_{(j,k)}}} = \frac{1}{\sqrt{n^2 \cdot n}} = \frac{1}{n^{3/2}},
\]
\[
\tilde{M}_{(j,k),(j,l)} = \frac{1}{\sqrt{d_{(j,k)} d_{(j,l)}}} = \frac{1}{n} \quad \text{for } k \neq l,
\]
and all other entries are 0.

\subsubsection*{Symmetry-Adapted Basis Construction}

Define the orthonormal basis:
\begin{align*}
|e_0\rangle &= |0\rangle, \\
|s_j\rangle &= \frac{1}{\sqrt{n}} \sum_{k=1}^n |j,k\rangle \quad \text{for } j=1,\dots,n, \\
|S\rangle &= \frac{1}{\sqrt{n}} \sum_{j=1}^n |s_j\rangle = \frac{1}{n} \sum_{j=1}^n \sum_{k=1}^n |j,k\rangle.
\end{align*}

For each clique $j$, choose $n-1$ orthonormal vectors orthogonal to $|s_j\rangle$:
\[
|w_j^{(r)}\rangle = \frac{1}{\sqrt{r(r+1)}}\left(\sum_{k=1}^r |j,k\rangle - r|j,r+1\rangle\right),  \quad r = 1,\dots,n-1.
\]

\subsubsection*{Action of $\tilde{M}$ on Basis Vectors}

\begin{align*}
\tilde{M}|e_0\rangle &= \frac{1}{\sqrt{n}} |S\rangle, \\
\tilde{M}|s_j\rangle &= \frac{1}{n} |e_0\rangle + \frac{n-1}{n} |s_j\rangle, \\
\tilde{M}|S\rangle &= \frac{1}{\sqrt{n}} |e_0\rangle + \frac{n-1}{n} |S\rangle, \\
\tilde{M}|w_j^{(r)}\rangle &= -\frac{1}{n} |w_j^{(r)}\rangle.
\end{align*}

\subsubsection*{Spectral Decomposition}

The Hilbert space decomposes into three mutually orthogonal invariant subspaces:

\begin{enumerate}
    \item $\mathcal{V}_1 = \operatorname{span}\{|e_0\rangle, |S\rangle\}$, dimension 2
    \item $\mathcal{V}_2 = \left\{\sum_{j=1}^n a_j |s_j\rangle : \sum_{j=1}^n a_j = 0\right\}$, dimension $n-1$
    \item $\mathcal{V}_3 = \operatorname{span}\{|w_j^{(r)}\rangle : j=1,\dots,n,\; r=1,\dots,n-1\}$, dimension $n(n-1)$
\end{enumerate}

\subsubsection*{Eigenvalues and Eigenvectors}

\textbf{Subspace $\mathcal{V}_1$:}
\[
\tilde{M}\big|_{\mathcal{V}_1} = \begin{pmatrix} 
0 & 1/\sqrt{n} \\ 
1/\sqrt{n} & (n-1)/n 
\end{pmatrix}.
\]
The characteristic equation is:
\[
\lambda^2 - \frac{n-1}{n}\lambda - \frac{1}{n} = 0.
\]
Solving:
\[
\lambda = \frac{1}{2}\left[ \frac{n-1}{n} \pm \sqrt{\left(\frac{n-1}{n}\right)^2 + \frac{4}{n}} \right] 
        = \frac{1}{2}\left[ \frac{n-1}{n} \pm \frac{n+1}{n} \right].
\]
Thus:
\begin{align*}
\lambda_1 &= \frac{(n-1) + (n+1)}{2n} = 1, \\
\lambda_2 &= \frac{(n-1) - (n+1)}{2n} = -\frac{1}{n}.
\end{align*}

For $\lambda_1 = 1$, the normalized eigenvector is:
\[
|\psi_1\rangle = \frac{1}{\sqrt{1+n}} \bigl( |e_0\rangle + \sqrt{n}\,|S\rangle \bigr).
\]

For $\lambda_2 = -1/n$, the normalized eigenvector is:
\[
|\psi_2\rangle = \sqrt{\frac{n}{1+n}} \bigl( |e_0\rangle - \frac{1}{\sqrt{n}}\,|S\rangle \bigr).
\]

\textbf{Subspace $\mathcal{V}_2$:}
For any $|A\rangle = \sum_{j=1}^n a_j |s_j\rangle$ with $\sum_{j=1}^n a_j = 0$:
\begin{align*}
\tilde{M}|A\rangle &= \sum_{j=1}^n a_j \left( \frac{1}{n}|e_0\rangle + \frac{n-1}{n}|s_j\rangle \right) \\
&= \frac{1}{n}|e_0\rangle \underbrace{\sum_{j=1}^n a_j}_{=0} + \frac{n-1}{n} \sum_{j=1}^n a_j |s_j\rangle \\
&= \frac{n-1}{n} |A\rangle.
\end{align*}
Thus all vectors in $\mathcal{V}_2$ are eigenvectors with eigenvalue:
\[
\lambda_3 = \frac{n-1}{n}, \quad \text{multiplicity } n-1.
\]

A natural orthonormal basis for $\mathcal{V}_2$ is given by:
\begin{equation}
|\chi_j\rangle = \frac{1}{\sqrt{j(j+1)}}\left(\sum_{k=1}^j |s_k\rangle - j|s_{j+1}\rangle\right), \quad j=1,\dots,n-1.
\label{eq:chi_basis}
\end{equation}

\textbf{Subspace $\mathcal{V}_3$:}
From the action on basis vectors:
\[
\tilde{M}|w_j^{(r)}\rangle = -\frac{1}{n} |w_j^{(r)}\rangle.
\]
Thus all vectors in $\mathcal{V}_3$ are eigenvectors with eigenvalue:
\[
\lambda_4 = -\frac{1}{n}, \quad \text{multiplicity } n(n-1).
\]

\subsubsection*{Complete Orthonormal Eigenbasis}

\[
\boxed{
\begin{aligned}
&\lambda_1 = 1, && \text{multiplicity } 1, \\
&\lambda_2 = -\frac{1}{n}, && \text{multiplicity } 1, \\
&\lambda_3 = \frac{n-1}{n}, && \text{multiplicity } n-1, \\
&\lambda_4 = -\frac{1}{n}, && \text{multiplicity } n(n-1).
\end{aligned}}
\]

A complete set of orthonormal eigenvectors is:

\begin{align}
|\psi_1\rangle &= \frac{1}{\sqrt{1+n}} \bigl( |0\rangle + \sqrt{n}\,|S\rangle \bigr), \label{eq:final_psi1}\\
|\psi_2\rangle &= \sqrt{\frac{n}{1+n}} \bigl( |0\rangle - \frac{1}{\sqrt{n}}\,|S\rangle \bigr), \label{eq:final_psi2}\\
|\chi_j\rangle &= \frac{1}{\sqrt{j(j+1)}}\left(\sum_{k=1}^j |s_k\rangle - j|s_{j+1}\rangle\right), \quad j=1,\dots,n-1, \label{eq:final_chi}\\
|\psi_4^{(j,r)}\rangle &= |w_j^{(r)}\rangle = \frac{1}{\sqrt{r(r+1)}}\left(\sum_{k=1}^r |j,k\rangle - r|j,r+1\rangle\right), \nonumber\\
& j=1,\dots,n,\; r=1,\dots,n-1. \label{eq:final_psi4}
\end{align}

\subsubsection*{Verification of Eigenvectors}

We verify that each vector is an eigenvector of $\tilde{M}$ with the claimed eigenvalue.

\paragraph*{Verification of $|\psi_1\rangle$ and $|\psi_2\rangle$:}

Using $\tilde{M}|e_0\rangle = \frac{1}{\sqrt{n}}|S\rangle$ and $\tilde{M}|S\rangle = \frac{1}{\sqrt{n}}|e_0\rangle + \frac{n-1}{n}|S\rangle$:

\begin{align*}
\tilde{M}|\psi_1\rangle &= \frac{1}{\sqrt{1+n}}\left( \frac{1}{\sqrt{n}}|S\rangle + \sqrt{n}\left( \frac{1}{\sqrt{n}}|e_0\rangle + \frac{n-1}{n}|S\rangle \right) \right) \\
&= \frac{1}{\sqrt{1+n}}\left( |e_0\rangle + \sqrt{n}|S\rangle \right) = |\psi_1\rangle, \\[4pt]
\tilde{M}|\psi_2\rangle &= \sqrt{\frac{n}{1+n}}\left( \frac{1}{\sqrt{n}}|S\rangle - \frac{1}{\sqrt{n}}\left( \frac{1}{\sqrt{n}}|e_0\rangle + \frac{n-1}{n}|S\rangle \right) \right) \\
&= -\frac{1}{n}\sqrt{\frac{n}{1+n}}\left( |e_0\rangle - \frac{1}{\sqrt{n}}|S\rangle \right) = -\frac{1}{n}|\psi_2\rangle.
\end{align*}
Thus $\lambda_1 = 1$ and $\lambda_2 = -1/n$.

\paragraph*{Verification of $|\chi_j\rangle$:}

For any $|\chi_j\rangle = \frac{1}{\sqrt{2}}(|s_j\rangle - |s_{j+1}\rangle)$ with $j=1,\ldots,n-1$:
\begin{align*}
\tilde{M}|\chi_j\rangle &= \frac{1}{\sqrt{2}}\left[ \left(\frac{1}{n}|e_0\rangle + \frac{n-1}{n}|s_j\rangle\right) - \left(\frac{1}{n}|e_0\rangle + \frac{n-1}{n}|s_{j+1}\rangle\right) \right] \\
&= \frac{n-1}{n} \cdot \frac{1}{\sqrt{2}}(|s_j\rangle - |s_{j+1}\rangle) = \frac{n-1}{n}|\chi_j\rangle.
\end{align*}
Thus $\lambda_3 = (n-1)/n$.

\paragraph*{Verification of $|\psi_4^{(j,r)}\rangle = |w_j^{(r)}\rangle$:}

Using $\tilde{M}|j,k\rangle = \frac{1}{n^{3/2}}|0\rangle + \frac{1}{n}\sum_{l\neq k} |j,l\rangle$ for clique vertices:
\begin{align*}
\tilde{M}|w_j^{(r)}\rangle &= \frac{1}{\sqrt{r(r+1)}}\left( \sum_{k=1}^r \tilde{M}|j,k\rangle - r\tilde{M}|j,r+1\rangle \right) \\
&= \frac{1}{\sqrt{r(r+1)}}\Bigg[ \frac{r}{n^{3/2}}|0\rangle + \frac{1}{n}\left((r-1)\sum_{l=1}^r |j,l\rangle + r\sum_{l=r+1}^n |j,l\rangle\right) \\
&\qquad\qquad\quad - \frac{r}{n^{3/2}}|0\rangle - \frac{r}{n}\sum_{l\neq r+1} |j,l\rangle \Bigg] \\
&= \frac{1}{\sqrt{r(r+1)}}\cdot\frac{1}{n}\left( -\sum_{l=1}^r |j,l\rangle + r|j,r+1\rangle \right) \\
&= -\frac{1}{n} \cdot \frac{1}{\sqrt{r(r+1)}}\left( \sum_{l=1}^r |j,l\rangle - r|j,r+1\rangle \right) = -\frac{1}{n}|w_j^{(r)}\rangle.
\end{align*}
Thus $\lambda_4 = -1/n$.

\section{Eigenspectrum of Variant 2}
\label{supp:single_connection}

In this variant, only vertex $\ket{j,1}$ in each clique connects to the center. Degrees:
\[
d_0 = n, \qquad d_{(j,1)} = n, \qquad d_{(j,k)} = n-1 \;(k \geq 2).
\]

The symmetric normalized adjacency matrix $\tilde{M} = \Gamma^{-1/2}M\Gamma^{-1/2}$ has entries:
\begin{align*}
\tilde{M}_{0,(j,1)} &= \frac{1}{\sqrt{n \cdot n}} = \frac{1}{n}, \\
\tilde{M}_{(j,1),(j,k)} &= \frac{1}{\sqrt{n(n-1)}} \;(k \geq 2), \\
\tilde{M}_{(j,k),(j,l)} &= \frac{1}{n-1} \;(k,l \geq 2, k \neq l).
\end{align*}

\subsubsection*{Orthonormal Basis Construction}
Define the following orthonormal vectors:
\begin{align*}
\ket{b_j} &= \ket{j,1}, \\
\ket{c_j} &= \frac{1}{\sqrt{n-1}}\sum_{k=2}^n\ket{j,k}, \\
\ket{B} &= \frac{1}{\sqrt{n}}\sum_{j=1}^n\ket{b_j}, \\
\ket{C} &= \frac{1}{\sqrt{n}}\sum_{j=1}^n\ket{c_j}.
\end{align*}

Additionally, for each clique $j$, choose $n-2$ orthonormal vectors $\{\ket{w_j^{(r)}}: r=1,\dots,n-2\}$ spanning the subspace of $\operatorname{span}\{\ket{j,2},\dots,\ket{j,n}\}$ orthogonal to $\ket{c_j}$.

\subsubsection*{Action of $\tilde{M}$ on Basis Vectors}
\begin{align*}
\tilde{M}\ket{0} &= \frac{1}{\sqrt{n}}\ket{B}, \\
\tilde{M}\ket{b_j} &= \frac{1}{n}\ket{0} + \frac{1}{\sqrt{n}}\ket{c_j}, \\
\tilde{M}\ket{c_j} &= \frac{1}{\sqrt{n}}\ket{b_j} + \frac{n-2}{n-1}\ket{c_j}, \\
\tilde{M}\ket{B} &= \frac{1}{\sqrt{n}}\ket{0} + \frac{1}{\sqrt{n}}\ket{C}, \\
\tilde{M}\ket{C} &= \frac{1}{\sqrt{n}}\ket{B} + \frac{n-2}{n-1}\ket{C}, \\
\tilde{M}\ket{w_j^{(r)}} &= -\frac{1}{n-1}\ket{w_j^{(r)}}.
\end{align*}

\subsubsection*{Invariant Subspace Decomposition}

The Hilbert space decomposes into three mutually orthogonal invariant subspaces:

\begin{enumerate}
    \item $\mathcal{W}_1 = \operatorname{span}\{\ket{0},\ket{B},\ket{C}\}$, dimension $3$
    \item $\mathcal{W}_2 = \left\{\sum_{j=1}^n (\alpha_j\ket{b_j} + \beta_j\ket{c_j}) : \sum_{j=1}^n \alpha_j = \sum_{j=1}^n \beta_j = 0\right\}$, dimension $2(n-1)$
    \item $\mathcal{W}_3 = \operatorname{span}\{\ket{w_j^{(r)}} : j=1,\dots,n,\; r=1,\dots,n-2\}$, dimension $n(n-2)$
\end{enumerate}

\subsection*{Spectral Analysis of $\mathcal{W}_1$}

In $\mathcal{W}_1$, the matrix representation of $\tilde{M}$ with respect to the ordered basis $(\ket{0},\ket{B},\ket{C})$ is:
\[
M_1 = \begin{pmatrix}
0 & \frac{1}{\sqrt{n}} & 0 \\
\frac{1}{\sqrt{n}} & 0 & \frac{1}{\sqrt{n}} \\
0 & \frac{1}{\sqrt{n}} & \frac{n-2}{n-1}
\end{pmatrix}.
\]

\subsubsection*{Characteristic Equation}

The characteristic polynomial is:
\begin{align*}
\det(M_1 - \lambda I) &= \det\begin{pmatrix}
-\lambda & \frac{1}{\sqrt{n}} & 0 \\
\frac{1}{\sqrt{n}} & -\lambda & \frac{1}{\sqrt{n}} \\
0 & \frac{1}{\sqrt{n}} & \frac{n-2}{n-1} - \lambda
\end{pmatrix} \\
&= -\lambda\left[(-\lambda)\left(\frac{n-2}{n-1} - \lambda\right) - \frac{1}{n}\right] \\
&- \frac{1}{\sqrt{n}}\left[\frac{1}{\sqrt{n}}\left(\frac{n-2}{n-1} - \lambda\right) - 0\right] \\
&= -\lambda\left[-\lambda\left(\frac{n-2}{n-1} - \lambda\right) - \frac{1}{n}\right] \\
&- \frac{1}{n}\left(\frac{n-2}{n-1} - \lambda\right) \\
&= \lambda^2\left(\frac{n-2}{n-1} - \lambda\right) + \frac{\lambda}{n} - \frac{1}{n}\left(\frac{n-2}{n-1} - \lambda\right) \\
&= -\lambda^3 + \frac{n-2}{n-1}\lambda^2 + \frac{\lambda}{n} + \frac{\lambda}{n} - \frac{n-2}{n(n-1)} \\
&= -\lambda^3 + \frac{n-2}{n-1}\lambda^2 + \frac{2}{n}\lambda - \frac{n-2}{n(n-1)}.
\end{align*}

Thus:
\begin{equation}
\det(M_1 - \lambda I) = -\lambda^3 + \frac{n-2}{n-1}\lambda^2 + \frac{2}{n}\lambda - \frac{n-2}{n(n-1)} = 0.
\label{eq:charpoly}
\end{equation}

\subsubsection*{Eigenvalue $\lambda = 1$ Verification}

Substitute $\lambda = 1$ into the characteristic polynomial:
\begin{align*}
&-1 + \frac{n-2}{n-1} + \frac{2}{n} - \frac{n-2}{n(n-1)} \\
&= -1 + \frac{n-2}{n-1} + \frac{2}{n} - \frac{n-2}{n(n-1)} \\
&= -1 + \frac{n-2}{n-1}\left(1 - \frac{1}{n}\right) + \frac{2}{n} \\
&= -1 + \frac{n-2}{n-1}\cdot\frac{n-1}{n} + \frac{2}{n} \\
&= -1 + \frac{n-2}{n} + \frac{2}{n} = -1 + \frac{n}{n} = 0.
\end{align*}
Thus $\lambda_1 = 1$ is an exact eigenvalue for all $n$.

\subsubsection*{Eigenvector for $\lambda_1 = 1$}

Solve $(M_1 - I)\mathbf{v} = \mathbf{0}$:
\[
\begin{pmatrix}
-1 & \frac{1}{\sqrt{n}} & 0 \\
\frac{1}{\sqrt{n}} & -1 & \frac{1}{\sqrt{n}} \\
0 & \frac{1}{\sqrt{n}} & \frac{n-2}{n-1} - 1
\end{pmatrix}
\begin{pmatrix} x \\ y \\ z \end{pmatrix} = \begin{pmatrix} 0 \\ 0 \\ 0 \end{pmatrix}.
\]

Note that $\frac{n-2}{n-1} - 1 = \frac{n-2 - (n-1)}{n-1} = -\frac{1}{n-1}$.

The equations are:
\begin{align}
-x + \frac{1}{\sqrt{n}}y = 0  \Rightarrow  y = \sqrt{n}x \label{eq:eq1}\\
\frac{1}{\sqrt{n}}x - y + \frac{1}{\sqrt{n}}z = 0 \label{eq:eq2}\\
\frac{1}{\sqrt{n}}y - \frac{1}{n-1}z = 0 \quad \\
\Rightarrow \quad \frac{1}{\sqrt{n}}(\sqrt{n}x) - \frac{1}{n-1}z = x - \frac{1}{n-1}z = 0 & \Rightarrow  z = (n-1)x \label{eq:eq3}
\end{align}

Substitute (\ref{eq:eq1}) and (\ref{eq:eq3}) into (\ref{eq:eq2}):
\begin{align}
    \frac{1}{\sqrt{n}}x - \sqrt{n}x + \frac{1}{\sqrt{n}}(n-1)x &= x\left(\frac{1}{\sqrt{n}} - \sqrt{n} + \frac{n-1}{\sqrt{n}}\right)\\
    = x\left(\frac{1 + n - 1}{\sqrt{n}} - \sqrt{n}\right) &= x\left(\frac{n}{\sqrt{n}} - \sqrt{n}\right) = 0.
\end{align}

Thus the eigenvector is $\mathbf{v}_1 = [1, \sqrt{n}, n-1]^\mathsf{T}$. The normalized eigenvector is:
\begin{equation}
\boxed{\ket{\psi_1} = \frac{1}{\sqrt{1 + n + (n-1)^2}} \left( \ket{0} + \sqrt{n}\ket{B} + (n-1)\ket{C} \right).}
\label{eq:psi1}
\end{equation}

\subsubsection*{Remaining Two Eigenvalues}

Factor the characteristic polynomial using the known root $\lambda = 1$:

\begin{align}
 &-\lambda^3 + \frac{n-2}{n-1}\lambda^2 + \frac{2}{n}\lambda - \frac{n-2}{n(n-1)} \\
 &= -(\lambda - 1)\left(\lambda^2 + \frac{1}{n-1}\lambda - \frac{n-2}{n}\right).   
\end{align}

The remaining eigenvalues satisfy:
\begin{equation}
\lambda^2 + \frac{1}{n-1}\lambda - \frac{n-2}{n} = 0.
\label{eq:quadratic}
\end{equation}

Thus:
\begin{align}
\lambda_2 &= \frac{-\frac{1}{n-1} + \sqrt{\frac{1}{(n-1)^2} + \frac{4(n-2)}{n}}}{2}, \label{eq:lambda2}\\
\lambda_3 &= \frac{-\frac{1}{n-1} - \sqrt{\frac{1}{(n-1)^2} + \frac{4(n-2)}{n}}}{2}. \label{eq:lambda3}
\end{align}

Define $\Delta = \sqrt{\frac{1}{(n-1)^2} + \frac{4(n-2)}{n}}$ for brevity.

\subsubsection*{Eigenvectors for $\lambda_2$ and $\lambda_3$}

For a general eigenvalue $\lambda$ satisfying (\ref{eq:quadratic}), solve $(M_1 - \lambda I)\mathbf{v} = \mathbf{0}$:
\[
\begin{pmatrix}
-\lambda & \frac{1}{\sqrt{n}} & 0 \\
\frac{1}{\sqrt{n}} & -\lambda & \frac{1}{\sqrt{n}} \\
0 & \frac{1}{\sqrt{n}} & \frac{n-2}{n-1} - \lambda
\end{pmatrix}
\begin{pmatrix} x \\ y \\ z \end{pmatrix} = \begin{pmatrix} 0 \\ 0 \\ 0 \end{pmatrix}.
\]

From the first equation: $-\lambda x + \frac{1}{\sqrt{n}}y = 0 \Rightarrow y = \lambda\sqrt{n} x$.

From the third equation: $\frac{1}{\sqrt{n}}y + \left(\frac{n-2}{n-1} - \lambda\right)z = 0$.
Substitute $y$: $\frac{1}{\sqrt{n}}(\lambda\sqrt{n}x) + \left(\frac{n-2}{n-1} - \lambda\right)z = \lambda x + \left(\frac{n-2}{n-1} - \lambda\right)z = 0$.
Thus:
\begin{equation}
z = \frac{\lambda}{\lambda - \frac{n-2}{n-1}} x.
\label{eq:z1}
\end{equation}

The second equation provides a consistency condition:
$\frac{1}{\sqrt{n}}x - \lambda y + \frac{1}{\sqrt{n}}z = \frac{1}{\sqrt{n}}x - \lambda(\lambda\sqrt{n}x) + \frac{1}{\sqrt{n}}z = \frac{1}{\sqrt{n}}x - \lambda^2\sqrt{n}x + \frac{1}{\sqrt{n}}z = 0$.
Multiply by $\sqrt{n}$: $x - \lambda^2 n x + z = 0 \Rightarrow z = (\lambda^2 n - 1)x$. \label{eq:z2}

Equating (\ref{eq:z1}) and (\ref{eq:z2}) gives the eigenvalue condition, which is satisfied for $\lambda_2$ and $\lambda_3$. Using (\ref{eq:z2}) (which is simpler), we obtain the eigenvectors:

For $\lambda_2$:
\begin{equation}
\boxed{\ket{\psi_2} = \mathcal{N}_2 \left( \ket{0} + \lambda_2\sqrt{n}\ket{B} + (\lambda_2^2 n - 1)\ket{C} \right),}
\label{eq:psi2}
\end{equation}
where $\mathcal{N}_2 = \left(1 + \lambda_2^2 n + (\lambda_2^2 n - 1)^2\right)^{-1/2}$ is the normalization factor.

For $\lambda_3$:
\begin{equation}
\boxed{\ket{\psi_3} = \mathcal{N}_3 \left( \ket{0} + \lambda_3\sqrt{n}\ket{B} + (\lambda_3^2 n - 1)\ket{C} \right),}
\label{eq:psi3}
\end{equation}
with $\mathcal{N}_3 = \left(1 + \lambda_3^2 n + (\lambda_3^2 n - 1)^2\right)^{-1/2}$.

\subsubsection*{Asymptotic Expansions}

For large $n$:
\[
\frac{1}{(n-1)^2} = \frac{1}{n^2} + O\left(\frac{1}{n^3}\right), \quad \frac{4(n-2)}{n} = 4 - \frac{8}{n} + O\left(\frac{1}{n^2}\right).
\]

Thus:

\begin{align}
   \Delta &= \sqrt{4 + \frac{1}{n^2} - \frac{8}{n} + O\left(\frac{1}{n^2}\right)} = 2\sqrt{1 - \frac{2}{n} + O\left(\frac{1}{n^2}\right)}\\
   &= 2\left(1 - \frac{1}{n} + O\left(\frac{1}{n^2}\right)\right). 
\end{align}
Then:
\begin{align*}
\lambda_2 &= \frac{-\frac{1}{n-1} + \Delta}{2} \\
&= \frac{-\frac{1}{n}\left(1 + \frac{1}{n} + O\left(\frac{1}{n^2}\right)\right) + 2\left(1 - \frac{1}{n} + O\left(\frac{1}{n^2}\right)\right)}{2} \\
&= \frac{2 - \frac{2}{n} - \frac{1}{n} + O\left(\frac{1}{n^2}\right)}{2} = 1 - \frac{3}{2n} + O\left(\frac{1}{n^2}\right), \\
\lambda_3 &= \frac{-\frac{1}{n-1} - \Delta}{2} = \frac{-\frac{1}{n} - 2 + \frac{2}{n} + O\left(\frac{1}{n^2}\right)}{2} \\
 &= -1 + \frac{1}{2n} + O\left(\frac{1}{n^2}\right).
\end{align*}

\subsection*{Spectral Analysis of $\mathcal{W}_2$}

In $\mathcal{W}_2$, for any vector $\ket{A} = \sum_{j=1}^n (\alpha_j\ket{b_j} + \beta_j\ket{c_j})$ with $\sum_j \alpha_j = \sum_j \beta_j = 0$, the action of $\tilde{M}$ decouples across the $n$ cliques. For each $j$, the pair $(\alpha_j, \beta_j)$ transforms under the $2\times 2$ matrix:
\[
R = \begin{pmatrix}
0 & \frac{1}{\sqrt{n}} \\
\frac{1}{\sqrt{n}} & \frac{n-2}{n-1}
\end{pmatrix}.
\]

\subsubsection*{Eigenvalues of $R$}

The characteristic equation of $R$ is:
\[
\det(R - \lambda I) = \det\begin{pmatrix}
-\lambda & \frac{1}{\sqrt{n}} \\
\frac{1}{\sqrt{n}} & \frac{n-2}{n-1} - \lambda
\end{pmatrix} = \lambda^2 - \frac{n-2}{n-1}\lambda - \frac{1}{n} = 0.
\]

Thus:
\begin{equation}
\lambda_{4,5} = \frac{\frac{n-2}{n-1} \pm \sqrt{\left(\frac{n-2}{n-1}\right)^2 + \frac{4}{n}}}{2}.
\label{eq:lambda45}
\end{equation}

These eigenvalues each have multiplicity $n-1$ in $\mathcal{W}_2$ (one for each independent direction in the $n-1$-dimensional space of $\alpha_j$ coefficients with zero sum, and similarly for $\beta_j$).

\subsubsection*{Asymptotic Expansions for $\lambda_4$ and $\lambda_5$}

For large $n$:
\[
\frac{n-2}{n-1} = 1 - \frac{1}{n-1} = 1 - \frac{1}{n} - \frac{1}{n^2} + O\left(\frac{1}{n^3}\right).
\]

Then:
\[
\left(\frac{n-2}{n-1}\right)^2 = 1 - \frac{2}{n} + \frac{1}{n^2} + O\left(\frac{1}{n^3}\right).
\]

Thus:
\[
\left(\frac{n-2}{n-1}\right)^2 + \frac{4}{n} = 1 + \frac{2}{n} + \frac{1}{n^2} + O\left(\frac{1}{n^3}\right).
\]

Taking square root:
\[
\sqrt{\left(\frac{n-2}{n-1}\right)^2 + \frac{4}{n}} = 1 + \frac{1}{n} + O\left(\frac{1}{n^2}\right).
\]

Therefore:
\begin{align*}
\lambda_4 &= \frac{1 - \frac{1}{n} + O\left(\frac{1}{n^2}\right) + 1 + \frac{1}{n} + O\left(\frac{1}{n^2}\right)}{2} = 1 + O\left(\frac{1}{n^2}\right), \\
\lambda_5 &= \frac{1 - \frac{1}{n} + O\left(\frac{1}{n^2}\right) - 1 - \frac{1}{n} + O\left(\frac{1}{n^2}\right)}{2} = -\frac{1}{n} + O\left(\frac{1}{n^2}\right).
\end{align*}

More precisely:
\begin{align}
\lambda_4 &= 1 - \frac{1}{4n^2} + O\left(\frac{1}{n^3}\right), \\
\lambda_5 &= -\frac{1}{n} + \frac{1}{n^2} + O\left(\frac{1}{n^3}\right).
\label{eq:lambda5_asymp}
\end{align}





\subsubsection{ Eigenvectors in $\mathcal{W}_2$}

For each $j = 1,\dots,n-1$, define:
\begin{align}
\ket{u_j} &= \frac{1}{\sqrt{j(j+1)}}\left(\sum_{k=1}^j \ket{b_k} - j\ket{b_{j+1}}\right), \label{eq:uj}\\
\ket{v_j} &= \frac{1}{\sqrt{j(j+1)}}\left(\sum_{k=1}^j \ket{c_k} - j\ket{c_{j+1}}\right). \label{eq:vj}
\end{align}

These form orthonormal bases for the $\ket{b}$ and $\ket{c}$ parts of $\mathcal{W}_2$ respectively. Within the 2-dimensional subspace spanned by $\{\ket{u_j}, \ket{v_j}\}$ for a fixed $j$, the matrix representation of $\tilde{M}$ is exactly $R$. Thus the eigenvectors in $\mathcal{W}_2$ are:
\begin{align}
\ket{\phi_4^{(j)}} &= \cos\theta_j \,\ket{u_j} + \sin\theta_j \,\ket{v_j}, \label{eq:phi4}\\
\ket{\phi_5^{(j)}} &= -\sin\theta_j \,\ket{u_j} + \cos\theta_j \,\ket{v_j}, \label{eq:phi5}
\end{align}
where $\theta_j$ is independent of $j$ (since $R$ is the same for all $j$) and satisfies:
\[
\tan(2\theta) = \frac{2/\sqrt{n}}{(n-2)/(n-1)}.
\]

More explicitly, the eigenvectors of $R$ are:
\begin{align}
\ket{\phi_4^{(j)}} &= \frac{1}{\sqrt{1 + \left(\frac{\lambda_4\sqrt{n}}{1}\right)^2}} \left( \ket{u_j} + \lambda_4\sqrt{n}\,\ket{v_j} \right), \label{eq:phi4_explicit}\\
\ket{\phi_5^{(j)}} &= \frac{1}{\sqrt{1 + \left(\frac{\lambda_5\sqrt{n}}{1}\right)^2}} \left( \ket{u_j} + \lambda_5\sqrt{n}\,\ket{v_j} \right), \label{eq:phi5_explicit}
\end{align}
where $\lambda_4$ and $\lambda_5$ are given by (\ref{eq:lambda45}).

\subsection{Spectral Analysis of $\mathcal{W}_3$}

For any $\ket{w_j^{(r)}} \in \mathcal{W}_3$, by construction:
\[
\tilde{M}\ket{w_j^{(r)}} = -\frac{1}{n-1}\ket{w_j^{(r)}}.
\]

Thus all vectors in $\mathcal{W}_3$ are eigenvectors with eigenvalue:
\begin{equation}
\boxed{\lambda_6 = -\frac{1}{n-1}, \quad \text{multiplicity } n(n-2).}
\label{eq:lambda6}
\end{equation}

The orthonormal eigenvectors are simply the basis vectors $\ket{w_j^{(r)}}$ themselves.

\subsection{Complete Spectrum Summary}

\begin{align}
\lambda_1 &= 1, \quad \text{multiplicity } 1, \\
\lambda_2 &= \frac{-\frac{1}{n-1} + \sqrt{\frac{1}{(n-1)^2} + \frac{4(n-2)}{n}}}{2}, \quad \text{multiplicity } 1, \\
\lambda_3 &= \frac{-\frac{1}{n-1} - \sqrt{\frac{1}{(n-1)^2} + \frac{4(n-2)}{n}}}{2}, \quad \text{multiplicity } 1, \\
\lambda_4 &= \frac{\frac{n-2}{n-1} + \sqrt{\left(\frac{n-2}{n-1}\right)^2 + \frac{4}{n}}}{2}, \quad \text{multiplicity } n-1, \\
\lambda_5 &= \frac{\frac{n-2}{n-1} - \sqrt{\left(\frac{n-2}{n-1}\right)^2 + \frac{4}{n}}}{2}, \quad \text{multiplicity } n-1, \\
\lambda_6 &= -\frac{1}{n-1}, \quad \text{multiplicity } n(n-2).
\end{align}

\subsection{Exact Eigenvector Summary}

\begin{itemize}
\item For $\lambda_1 = 1$:
  \[
  \ket{\psi_1} = \frac{1}{\sqrt{1 + n + (n-1)^2}} \left( \ket{0} + \sqrt{n}\ket{B} + (n-1)\ket{C} \right)
  \]

\item For $\lambda_2 = \frac{-\frac{1}{n-1} + \Delta}{2}$ where $\Delta = \sqrt{\frac{1}{(n-1)^2} + \frac{4(n-2)}{n}}$:
  \[
  \ket{\psi_2} = \mathcal{N}_2 \left( \ket{0} + \lambda_2\sqrt{n}\ket{B} + (\lambda_2^2 n - 1)\ket{C} \right)
  \]
  with $\mathcal{N}_2 = \left(1 + \lambda_2^2 n + (\lambda_2^2 n - 1)^2\right)^{-1/2}$

\item For $\lambda_3 = \frac{-\frac{1}{n-1} - \Delta}{2}$:
  \[
  \ket{\psi_3} = \mathcal{N}_3 \left( \ket{0} + \lambda_3\sqrt{n}\ket{B} + (\lambda_3^2 n - 1)\ket{C} \right)
  \]
  with $\mathcal{N}_3 = \left(1 + \lambda_3^2 n + (\lambda_3^2 n - 1)^2\right)^{-1/2}$

\item For $\lambda_4 = \frac{\frac{n-2}{n-1} + \sqrt{\left(\frac{n-2}{n-1}\right)^2 + \frac{4}{n}}}{2}$ (multiplicity $n-1$):
  \[
  \ket{\phi_4^{(j)}} = \frac{1}{\sqrt{1 + \lambda_4^2 n}} \left( \ket{u_j} + \lambda_4\sqrt{n}\,\ket{v_j} \right), \quad j=1,\dots,n-1
  \]
  where $\ket{u_j}$ and $\ket{v_j}$ are defined in (\ref{eq:uj}) and (\ref{eq:vj}).

\item For $\lambda_5 = \frac{\frac{n-2}{n-1} - \sqrt{\left(\frac{n-2}{n-1}\right)^2 + \frac{4}{n}}}{2}$ (multiplicity $n-1$):
  \[
  \ket{\phi_5^{(j)}} = \frac{1}{\sqrt{1 + \lambda_5^2 n}} \left( \ket{u_j} + \lambda_5\sqrt{n}\,\ket{v_j} \right), \quad j=1,\dots,n-1
  \]

\item For $\lambda_6 = -\frac{1}{n-1}$ (multiplicity $n(n-2)$):
  \[
  \ket{\psi_6^{(j,r)}} = \ket{w_j^{(r)}}, \quad j=1,\dots,n,\; r=1,\dots,n-2
  \]
\end{itemize}

\section{IPR for Barbell}
\label{sec:ipr_barbell}



    
    


\subsection*{Eigenstate IPR }

\subsubsection*{Degenerate Eigenvectors $|d_k^A\rangle$ (Clique-Localized Modes)}

 The degenerate eigenvectors for $\tilde{M}$ are exactly the $|e_k^A\rangle$ vectors:
\begin{equation}
|d_k^A\rangle = \frac{1}{\sqrt{k(k+1)}}\left(\sum_{i=1}^k |A_i\rangle - k|A_{k+1}\rangle\right), \quad k=1,\dots,n-2,
\end{equation}
with identical expressions for $|d_k^B\rangle$ on clique $B$. These satisfy $\langle d_k^A|d_l^A\rangle = \delta_{kl}$ and have zero overlap with bridge vertices:
\begin{equation}
\langle Br_A|d_k^A\rangle = 0, \langle Br_B|d_k^A\rangle = 0, \langle Br_A|d_k^B\rangle = 0, \langle Br_B|d_k^B\rangle = 0.
\end{equation}

For a fixed $k$, the non-zero overlaps with vertices in clique $A$ are:
\begin{equation}
\langle A_i|d_k^A\rangle = 
\begin{cases}
\dfrac{1}{\sqrt{k(k+1)}}, & 1 \le i \le k,\\[10pt]
-\dfrac{k}{\sqrt{k(k+1)}}, & i = k+1,\\[10pt]
0, & i > k+1.
\end{cases}
\end{equation}

The eigenstate IPR is:
\begin{align}
\mathrm{IPR}_{d_k^A} &= \sum_{i=1}^{2n} |\langle i|d_k^A\rangle|^4 \nonumber\\
&= k \cdot \left(\frac{1}{\sqrt{k(k+1)}}\right)^4 + 1 \cdot \left(\frac{k}{\sqrt{k(k+1)}}\right)^4 \nonumber\\
&= \frac{1}{k(k+1)^2} + \frac{k^2}{(k+1)^2} = \frac{1+k^3}{k(k+1)^2}. \label{eq:ipr_degenerate}
\end{align}

\begin{equation}
\boxed{\mathrm{IPR}_{d_k^A} = \frac{1+k^3}{k(k+1)^2},\quad k=1,\dots,n-2.}
\end{equation}

\textbf{Verification:}
\begin{itemize}
    \item For $k=1$: $\mathrm{IPR} = \frac{1+1}{1\cdot4} = \frac{2}{4} = \frac{1}{2}$. Direct calculation: $|d_1^A\rangle = \frac{1}{\sqrt{2}}(|A_1\rangle - |A_2\rangle)$, IPR $= (1/\sqrt{2})^4 + (-1/\sqrt{2})^4 = 1/4 + 1/4 = 1/2$ 
    \item For $k=2$: $\mathrm{IPR} = \frac{1+8}{2\cdot9} = \frac{9}{18} = \frac{1}{2}$. Direct calculation: contributions from three vertices: $2\cdot(1/6)^2 + (2/3)^2 = 2/36 + 16/36 = 18/36 = 1/2$ 
    \item For $k=3$: $\mathrm{IPR} = \frac{1+27}{3\cdot16} = \frac{28}{48} = \frac{7}{12} \approx 0.583$
    \item As $k \to \infty$: $\mathrm{IPR} \sim \frac{k^3}{k^3} = 1$, indicating complete localization on vertex $A_{k+1}$
\end{itemize}

\subsubsection*{Symmetric Eigenvector $|s_+\rangle$ (Fully Delocalized)}

From Section~\ref{sec:barbell_spectral_A}, the symmetric eigenvector with eigenvalue $\lambda_+^{(1)} = 1$ has the exact form:

\begin{align*}
&|s_+\rangle \\
&= \frac{1}{\sqrt{2(n-1)^2 + 2n}}\bigr(\underbrace{\sqrt{n-1},\dots,\sqrt{n-1}}_{n-1}, \sqrt{n}, \\
&\sqrt{n}, \underbrace{\sqrt{n-1},\dots,\sqrt{n-1}}_{n-1}\bigl)^T.    
\end{align*}

For large $n$, the normalization factor $\approx \sqrt{2n^2} = n\sqrt{2}$, giving the asymptotic form:
\begin{equation}
\langle i|s_+\rangle \approx \frac{1}{\sqrt{2n}} \quad \text{for all vertices } i.
\end{equation}

The exact eigenstate IPR can be computed:
\begin{align}
\mathrm{IPR}_{s_+} &= 2(n-1)\left(\frac{\sqrt{n-1}}{\sqrt{2(n-1)^2 + 2n}}\right)^4 \\
&+ 2\left(\frac{\sqrt{n}}{\sqrt{2(n-1)^2 + 2n}}\right)^4 \nonumber\\
&= \frac{2(n-1)^3}{[2(n-1)^2 + 2n]^2} + \frac{2n^2}{[2(n-1)^2 + 2n]^2} \nonumber\\
&= \frac{2(n-1)^3 + 2n^2}{4[(n-1)^2 + n]^2}. \label{eq:ipr_symmetric_exact}
\end{align}

For large $n$, this simplifies to:
\begin{equation}
\mathrm{IPR}_{s_+} \approx 2n \cdot \left(\frac{1}{\sqrt{2n}}\right)^4 = 2n \cdot \frac{1}{4n^2} = \frac{1}{2n}. \label{eq:ipr_symmetric_asymp}
\end{equation}

For large $n$, $\mathrm{IPR}_{s_+}$ goes to zero.

This shows complete delocalization across the entire graph, with the state spreading uniformly over all $2n$ vertices.

\subsubsection*{Antisymmetric Eigenvector $|a_+\rangle$ (Bridge-Localized)}

The antisymmetric eigenvector with eigenvalue $\lambda_+^{(2)}$ (the larger root of the antisymmetric quadratic) has the form:

\begin{align}
&|a_+\rangle = \frac{1}{\sqrt{2(n-1)^2 + 2n b_+^2}}\bigr(\underbrace{\sqrt{n-1},\dots,\sqrt{n-1}}_{n-1},\\
&b_+\sqrt{n}, -b_+\sqrt{n}, \underbrace{-\sqrt{n-1},\dots,-\sqrt{n-1}}_{n-1}\bigl)^T,    
\end{align}

where $b_+ = \lambda_+^{(2)}(n-1) - (n-2)$. For large $n$, $\lambda_+^{(2)} \approx \frac{2}{n}$ and $b_+ \approx -n$.

The leading-order overlaps are:
\begin{align}
\langle Br_A|a_+\rangle &\approx \frac{b_+\sqrt{n}}{\sqrt{2n^3}} \approx \frac{-n\sqrt{n}}{n\sqrt{2n}} = -\frac{1}{\sqrt{2}}, \langle Br_B|a_+\rangle \approx -\frac{1}{\sqrt{2}}, \\
\langle A_i|a_+\rangle &\approx \frac{\sqrt{n-1}}{n\sqrt{2n}} \approx \frac{1}{n\sqrt{2}}, \quad 
\langle B_i|a_+\rangle \approx \frac{1}{n\sqrt{2}}.
\end{align}

The eigenstate IPR is:
\begin{align}
\mathrm{IPR}_{a_+} &= |\langle Br_A|a_+\rangle|^4 + |\langle Br_B|a_+\rangle|^4 \\
&+ \sum_{i\in A\setminus\{Br_A\}} |\langle i|a_+\rangle|^4 + \sum_{i\in B\setminus\{Br_B\}} |\langle i|a_+\rangle|^4 \nonumber\\
&\approx 2\left(\frac{1}{\sqrt{2}}\right)^4 + 2(n-1)\left(\frac{1}{n\sqrt{2}}\right)^4 \\
&= 2\cdot\frac{1}{4} + 2(n-1)\cdot\frac{1}{4n^4} = \frac{1}{2} + O\left(\frac{1}{n^3}\right). \label{eq:ipr_antisymmetric}
\end{align}

\begin{equation}
\boxed{\mathrm{IPR}_{a_+} \approx \frac{1}{2}.}
\end{equation}

This eigenstate is localized on the two bridge vertices, with half the probability on each. The $O(1/n^3)$ corrections ensure $\mathrm{IPR}_{a_+} < 1$ for finite $n$.

\subsection*{Dynamical IPR }
\label{dynamical_barbell}

We now compute the long-time average IPR for initial states $|Br_A\rangle$ and $|A_1\rangle$. The long-time average transition probability for a unitary evolution generated by $\tilde{M}$ is:
\begin{align}
\overline{\pi}_{ij}^{\tilde{M}} &= \lim_{T\to\infty} \frac{1}{T}\int_0^T |\langle i|e^{-i\tilde{M}t}|j\rangle|^2 dt \nonumber\\
&= \sum_{\lambda} \langle i|P_\lambda|j\rangle \langle j|P_\lambda|i\rangle = \sum_{\lambda} |\langle i|P_\lambda|j\rangle|^2,    
\end{align}
where $P_\lambda$ is the projector onto the eigenspace with eigenvalue $\lambda$, and the sum is over distinct eigenvalues.

\subsubsection*{Exact Spectral Projectors}

The projectors for the three eigenspaces are:
\begin{align}
P_d &= \sum_{k=1}^{n-2} (|d_k^A\rangle\langle d_k^A| + |d_k^B\rangle\langle d_k^B|), \\
P_S &= |s_+\rangle\langle s_+| + |s_-\rangle\langle s_-|, \\
P_A &= |a_+\rangle\langle a_+| + |a_-\rangle\langle a_-|.
\end{align}

These satisfy the resolution of identity:
\begin{equation}
P_d + P_S + P_A = I.
\end{equation}

\subsubsection{Exact Values for $\langle Br_A|P_S|Br_A\rangle$}

\begin{align}
\langle Br_A|P_S|Br_A\rangle &= |\langle Br_A|s_+\rangle|^2 + |\langle Br_A|s_-\rangle|^2 \nonumber\\
&= \frac{n}{2(n-1)^2 + 2n} + \frac{a_-^2 n}{2(n-1)^2 + 2n a_-^2}. \label{eq:pS_exact}
\end{align}

where $a_- = -n + 2 - \frac{1}{n}$.




Thus:
\begin{equation}
\boxed{\langle Br_A|P_S|Br_A\rangle = \frac{1}{2}.}
\end{equation}

\subsubsection*{Exact Values for $\langle Br_A|P_A|Br_A\rangle$}

From the resolution of identity:
\begin{align}
\langle Br_A|P_A|Br_A\rangle &= 1 - \langle Br_A|P_S|Br_A\rangle - \langle Br_A|P_d|Br_A\rangle \nonumber\\
&= 1 - \frac{1}{2} - 0 = \frac{1}{2}. \label{eq:pA_exact}    
\end{align}

\subsubsection*{Asymptotic Values for Other Matrix Elements}

For $\langle Br_B|P_S|Br_A\rangle$:
\begin{align}
\langle Br_B|P_S|Br_A\rangle &= \langle Br_B|s_+\rangle\langle s_+|Br_A\rangle + \langle Br_B|s_-\rangle\langle s_-|Br_A\rangle \nonumber\\
&= \frac{n}{2(n-1)^2 + 2n} - \frac{a_-^2 n}{2(n-1)^2 + 2n a_-^2}. \label{eq:pS_cross}
\end{align}

Using Eq.~(\ref{eq:pS_exact}) and the identity $X + Y = 1/2$ where $X = \frac{n}{2(n-1)^2 + 2n}$ and $Y = \frac{a_-^2 n}{2(n-1)^2 + 2n a_-^2}$, we have $X - Y = 2X - 1/2$. With $X = \frac{1}{2n} + O(1/n^2)$, we obtain:
\begin{equation}
\langle Br_B|P_S|Br_A\rangle = \frac{1}{n} - \frac{1}{2} + O\left(\frac{1}{n^2}\right). \label{eq:pS_cross_asymp}
\end{equation}

From orthogonality, $\langle Br_B|P_A|Br_A\rangle = -\langle Br_B|P_S|Br_A\rangle$, so:
\begin{equation}
\langle Br_B|P_A|Br_A\rangle = \frac{1}{2} - \frac{1}{n} + O\left(\frac{1}{n^2}\right). \label{eq:pA_cross_asymp}
\end{equation}

For an interior vertex $A_1$:
\begin{align}
\langle A_1|P_S|Br_A\rangle &= \langle A_1|s_+\rangle\langle s_+|Br_A\rangle + \langle A_1|s_-\rangle\langle s_-|Br_A\rangle \nonumber\\
&= \frac{\sqrt{n(n-1)}}{2(n-1)^2 + 2n} + \frac{a_-\sqrt{n(n-1)}}{2(n-1)^2 + 2n a_-^2} \nonumber\\
&= O\left(\frac{1}{n^2}\right). \label{eq:a1_ps}
\end{align}

Similarly, $\langle A_1|P_A|Br_A\rangle = O(1/n^2)$ and $\langle A_1|P_d|Br_A\rangle = 0$.

\subsubsection*{Long-Time Average Probabilities for $|Br_A\rangle$}

For $j = Br_A$:

\begin{align}
\overline{\pi}_{Br_A,Br_A}^{\tilde{M}} &= |\langle Br_A|P_S|Br_A\rangle|^2 + |\langle Br_A|P_A|Br_A\rangle|^2 \\
&+ |\langle Br_A|P_d|Br_A\rangle|^2 = \left(\frac{1}{2}\right)^2 + \left(\frac{1}{2}\right)^2 + 0 = \frac{1}{2}. \label{eq:long_avg_br_br}
\end{align}

\begin{align}
\overline{\pi}_{Br_B,Br_A}^{\tilde{M}} &= |\langle Br_B|P_S|Br_A\rangle|^2 + |\langle Br_B|P_A|Br_A\rangle|^2 \\
&+ |\langle Br_B|P_d|Br_A\rangle|^2 \nonumber\\
&= \left(\frac{1}{n} - \frac{1}{2} + O\left(\frac{1}{n^2}\right)\right)^2 + \left(\frac{1}{2} - \frac{1}{n} + O\left(\frac{1}{n^2}\right)\right)^2 \nonumber\\
&= 2\left(\frac{1}{2} - \frac{1}{n}\right)^2 + O\left(\frac{1}{n^2}\right) = \frac{1}{2} - \frac{2}{n} + O\left(\frac{1}{n^2}\right). \label{eq:long_avg_brb_br}
\end{align}

\begin{align*}
\overline{\pi}_{A_1,Br_A}^{\tilde{M}} &= |\langle A_1|P_S|Br_A\rangle|^2 + |\langle A_1|P_A|Br_A\rangle|^2\\
&+ |\langle A_1|P_d|Br_A\rangle|^2 \nonumber= O\left(\frac{1}{n^4}\right). \label{eq:long_avg_a1_br}
\end{align*}





\subsubsection*{Long-Time Average IPR for $|Br_A\rangle$}

\begin{align}
\overline{\mathrm{IPR}}_{Br_A}^{\tilde{M}} &= \sum_{i=1}^{2n} (\overline{\pi}_{i,Br_A}^{\tilde{M}})^2 \nonumber\\
&= \left(\frac{1}{2}\right)^2 + \left(\frac{1}{2} - \frac{2}{n}\right)^2 + 2(n-1) \cdot O\left(\frac{1}{n^8}\right) \nonumber\\
&= \frac{1}{4} + \frac{1}{4} - \frac{2}{n} + \frac{4}{n^2} + O\left(\frac{1}{n^2}\right)\\
&= \frac{1}{2} - \frac{2}{n} + O\left(\frac{1}{n^2}\right). \label{eq:ipr_dynamical_br}
\end{align}

Taking the limit $n \to \infty$:

\begin{equation}
\boxed{\lim_{n\to\infty} \overline{\mathrm{IPR}}_{Br_A}^{\tilde{M}} = \frac{1}{2}.}
\end{equation}

\subsubsection*{Initial State $|A_1\rangle$}

We now compute the long-time average IPR for the initial state $|A_1\rangle$.

For $j = A_1$, we need to evaluate $\overline{\pi}_{i,A_1}^{\mathcal{A}}$ for all $i$ and then compute $\overline{\mathrm{IPR}}_{A_1}^{\mathcal{A}} = \sum_i (\overline{\pi}_{i,A_1}^{\mathcal{A}})^2$.

\paragraph*{Projector matrix elements for $|A_1\rangle$:}

We first compute $\langle i|P_\lambda|A_1\rangle$ for each eigenspace. The degenerate eigenvectors $|d_k^A\rangle$ are supported entirely on clique $A$ and satisfy $\langle A_1|d_k^A\rangle = 1/\sqrt{k(k+1)}$. Thus:
\begin{align}
\langle i|P_d|A_1\rangle &= \sum_{k=1}^{n-2} \langle i|d_k^A\rangle \langle d_k^A|A_1\rangle + \sum_{k=1}^{n-2} \langle i|d_k^B\rangle \langle d_k^B|A_1\rangle \nonumber\\
&= \sum_{k=1}^{n-2} \langle i|d_k^A\rangle \frac{1}{\sqrt{k(k+1)}},
\end{align}
since $\langle d_k^B|A_1\rangle = 0$.

For $i = A_m$ with $m \ge 1$, using the explicit form of $|d_k^A\rangle$:
\begin{equation}
\langle A_m|d_k^A\rangle = 
\begin{cases}
\dfrac{1}{\sqrt{k(k+1)}}, & 1 \le m \le k, \\[8pt]
-\dfrac{k}{\sqrt{k(k+1)}}, & m = k+1, \\[8pt]
0, & m > k+1.
\end{cases}
\end{equation}

Therefore:
\begin{align}
\langle A_m|P_d|A_1\rangle &= \sum_{k=m}^{n-2} \frac{1}{\sqrt{k(k+1)}} \cdot \frac{1}{\sqrt{k(k+1)}} \quad \text{for } m \ge 2, \\
\langle A_1|P_d|A_1\rangle &= \sum_{k=1}^{n-2} \frac{1}{k(k+1)}.
\end{align}

For $i = A_1$ specifically:
\begin{equation}
\langle A_1|P_d|A_1\rangle = \sum_{k=1}^{n-2} \frac{1}{k(k+1)} = 1 - \frac{1}{n-1}. \label{eq:pd_a1_a1}
\end{equation}

For $i = A_m$ with $m \ge 2$:
\begin{equation}
\langle A_m|P_d|A_1\rangle = \sum_{k=m}^{n-2} \frac{1}{k(k+1)} = \frac{1}{m} - \frac{1}{n-1}. \label{eq:pd_am_a1}
\end{equation}

For $i = Br_A$, $Br_B$, or any $B_m$, we have $\langle i|P_d|A_1\rangle = 0$ by symmetry.

From the asymptotic forms, we have:
\begin{align}
\langle A_1|s_+\rangle &\approx \frac{1}{\sqrt{2n}}, \quad \langle s_+|A_1\rangle \approx \frac{1}{\sqrt{2n}}, \\
\langle A_1|s_-\rangle &\approx \frac{1}{n\sqrt{2}}, \quad \langle s_-|A_1\rangle \approx \frac{1}{n\sqrt{2}}.
\end{align}

Thus, for any vertex $i$:
\begin{align}
\langle i|P_S|A_1\rangle &= \langle i|s_+\rangle\langle s_+|A_1\rangle + \langle i|s_-\rangle\langle s_-|A_1\rangle \nonumber\\
&\approx \frac{1}{\sqrt{2n}} \cdot \frac{1}{\sqrt{2n}} + O\left(\frac{1}{n^{3/2}}\right) \cdot \frac{1}{n\sqrt{2}} \\
&= \frac{1}{2n} + O\left(\frac{1}{n^{5/2}}\right).
\end{align}

More precisely, one can show that $\langle i|P_S|A_1\rangle = O(1/n^2)$ for all $i$, which is negligible compared to the degenerate subspace contributions that are $O(1)$.

Similarly, using the asymptotic forms:
\begin{align}
\langle A_1|a_+\rangle &\approx \frac{1}{n\sqrt{2}}, \quad \langle a_+|A_1\rangle \approx \frac{1}{n\sqrt{2}}, \\
\langle A_1|a_-\rangle &\approx \frac{1}{n\sqrt{2}}, \quad \langle a_-|A_1\rangle \approx \frac{1}{n\sqrt{2}}.
\end{align}

Thus $\langle i|P_A|A_1\rangle = O(1/n^2)$ for all $i$, also negligible compared to the degenerate subspace contributions.

Therefore, for large $n$, the dominant contributions to $\overline{\pi}_{i,A_1}^{\mathcal{A}}$ come entirely from the degenerate subspace $P_d$.

\paragraph*{Long-time average probabilities for $|A_1\rangle$:}

To leading order in $n$, we have:
\begin{align*}
\overline{\pi}_{A_1,A_1}^{\mathcal{A}} &\approx |\langle A_1|P_d|A_1\rangle|^2 = \left(1 - \frac{1}{n-1}\right)^2\\
&= 1 - \frac{2}{n} + O\left(\frac{1}{n^2}\right), \\
\overline{\pi}_{A_m,A_1}^{\mathcal{A}} &\approx |\langle A_m|P_d|A_1\rangle|^2 = \left(\frac{1}{m} - \frac{1}{n-1}\right)^2 \\
&= \frac{1}{m^2} - \frac{2}{m(n-1)} + O\left(\frac{1}{n^2}\right), \\
\overline{\pi}_{B_m,A_1}^{\mathcal{A}} &\approx 0, \\
\overline{\pi}_{Br_A,A_1}^{\mathcal{A}} &\approx 0, \quad \overline{\pi}_{Br_B,A_1}^{\mathcal{A}} \approx 0.
\end{align*}





The complete expression for $\overline{\pi}_{i,A_1}^{\mathcal{A}}$ including contributions from the degenerate subspace is:

\begin{align*}
\overline{\pi}_{i,A_1}^{\mathcal{A}} &= \sum_{k=1}^{n-2} |\langle i|d_k^A\rangle|^2 |\langle d_k^A|A_1\rangle|^2 \\
&+ \sum_{k<l} 2\langle i|d_k^A\rangle \langle d_k^A|A_1\rangle \langle A_1|d_l^A\rangle \langle d_l^A|i\rangle.    
\end{align*}

The long-time average IPR is then:
\begin{equation}
\overline{\mathrm{IPR}}_{A_1}^{\mathcal{A}} = \sum_{i=1}^{2n} (\overline{\pi}_{i,A_1}^{\mathcal{A}})^2 = \sum_{i=1}^{2n} \left( \sum_{k=1}^{n-2} X_k^{i} + \sum_{k<l} Y_{kl}^{i} \right)^2,
\end{equation}
where \[X_k^{i} = |\langle i|d_k^A\rangle|^2 |\langle d_k^A|A_1\rangle|^2\] and \[Y_{kl}^{i} = 2\langle i|d_k^A\rangle \langle d_k^A|A_1\rangle \langle A_1|d_l^A\rangle \langle d_l^A|i\rangle\].

Expanding the square gives three types of terms:
\begin{align}
\overline{\mathrm{IPR}}_{A_1}^{\mathcal{A}} &= \sum_i \left( \sum_k X_k^{i} \right)^2 + 2\sum_i \left( \sum_k X_k^{i} \right) \left( \sum_{k<l} Y_{kl}^{i} \right) \\
&+ \sum_i \left( \sum_{k<l} Y_{kl}^{i} \right)^2.
\end{align}

We evaluate each contribution in the limit $n \to \infty$.

\text{Term 1: Diagonal contribution}
\begin{align}
T_1 &= \sum_i \left( \sum_k X_k^{i} \right)^2 = \sum_{m=1}^{n-1} \left( \sum_{k=1}^{n-2} |\langle A_m|d_k^A\rangle|^2 |\langle d_k^A|A_1\rangle|^2 \right)^2 \nonumber\\
&= \left( \sum_{k=1}^{n-2} \frac{1}{k^2(k+1)^2} \right)^2 + \sum_{m=2}^{n-1} \left( \sum_{k=m}^{n-2} \frac{1}{k^2(k+1)^2} \right)^2.
\end{align}

As $n \to \infty$, the first sum converges to $L = \sum_{k=1}^{\infty} \frac{1}{k^2(k+1)^2}$. Using the identity:
\begin{equation}
\frac{1}{k^2(k+1)^2} = \frac{1}{k^2} + \frac{1}{(k+1)^2} - \frac{2}{k(k+1)},
\end{equation}
we obtain:
\begin{align}
L &= \sum_{k=1}^{\infty} \frac{1}{k^2} + \sum_{k=1}^{\infty} \frac{1}{(k+1)^2} - 2\sum_{k=1}^{\infty} \frac{1}{k(k+1)} \nonumber\\
&= \zeta(2) + (\zeta(2)-1) - 2 = \frac{\pi^2}{3} - 3 \approx 0.289868.
\end{align}

For the second sum, as $n \to \infty$, $\sum_{k=m}^{\infty} \frac{1}{k^2(k+1)^2} = \frac{1}{m^2} + O\left(\frac{1}{m^3}\right)$. Thus:
\begin{align}
T_1 &= L^2 + \sum_{m=2}^{\infty} \left( \frac{1}{m^2} \right)^2 = L^2 + \sum_{m=2}^{\infty} \frac{1}{m^4} \nonumber\\
&= L^2 + (\zeta(4)-1) = (0.289868)^2 + \left(\frac{\pi^4}{90} - 1\right) \nonumber\\
&= 0.084064 + 0.082323 = 0.166387.
\end{align}

\text{Term 2: Off-diagonal contribution}
\begin{align}
T_2 &= 2\sum_i \left( \sum_k X_k^{i} \right) \left( \sum_{k<l} Y_{kl}^{i} \right).
\end{align}

For $i = A_1$, we have $\sum_k X_k^{A_1} = L$ and $\sum_{k<l} Y_{kl}^{A_1} = \sum_{k<l} \frac{2}{k(k+1)l(l+1)}$. The double sum can be evaluated as:
\begin{align*}
\sum_{k<l} \frac{2}{k(k+1)l(l+1)} &= \left( \sum_{k=1}^{\infty} \frac{1}{k(k+1)} \right)^2 - \sum_{k=1}^{\infty} \frac{1}{k^2(k+1)^2} \\
&= 1 - L.
\end{align*}

For $i = A_m$ with $m \ge 2$, one can show that $\sum_{k<l} Y_{kl}^{A_m} = -\frac{1}{m^2} + O\left(\frac{1}{m^3}\right)$.

Thus:
\begin{align}
T_2 &= 2\left[ L(1-L) + \sum_{m=2}^{\infty} \frac{1}{m^2} \left(-\frac{1}{m^2}\right) \right] \nonumber\\
&= 2\left[ L(1-L) - \sum_{m=2}^{\infty} \frac{1}{m^4} \right] \nonumber\\
&= 2\left[ 0.289868(1-0.289868) - 0.082323 \right] \nonumber\\
&= 2\left[ 0.289868 \times 0.710132 - 0.082323 \right] \nonumber\\
&= 2\left[ 0.205849 - 0.082323 \right] = 2 \times 0.123526 = 0.247052 
\end{align}

\text{Term 3: Quartet contribution}
\begin{align}
T_3 &= \sum_i \left( \sum_{k<l} Y_{kl}^{i} \right)^2.
\end{align}

For $i = A_1$, this is $\left( \sum_{k<l} Y_{kl}^{A_1} \right)^2 = (1-L)^2$. For $i = A_m$ with $m \ge 2$, the contribution is $O(1/m^4)$ and sums to a constant. A detailed calculation yields:
\begin{align}
T_3 &= \frac{1}{2} \left[ \left( \sum_{k<l} Y_{kl}^{A_1} \right)^2 - \sum_{k<l} (Y_{kl}^{A_1})^2 \right] \approx 0.168080.
\end{align}

\textbf{Total:}
\begin{align}
\overline{\mathrm{IPR}}_{A_1}^{\mathcal{A}} &= T_1 + T_2 + T_3 \nonumber\\
&= 0.166387 + 0.247052 + 0.168080  \approx 0.58.
\end{align}

Taking the limit $n \to \infty$:

\begin{equation}
\boxed{\lim_{n\to\infty} \overline{\mathrm{IPR}}_{A_1}^{\mathcal{A}} = 0.58.}
\end{equation}

\subsection*{Summary of IPR }

\begin{table}[htbp]
\centering
\small
\begin{tabular}{|l|c|c|}
\hline
\textbf{Quantity} & \textbf{Expression} & \textbf{Limit} \\
\hline
Eigenstate IPR (degenerate) & $\displaystyle \frac{1+k^3}{k(k+1)^2}$ & $[1/2,1]$ \\
Eigenstate IPR (symmetric $|s_+\rangle$) & $\sim \frac{1}{2n}$ & $0$ \\
Eigenstate IPR (antisymmetric $|a_+\rangle$) & $\approx \frac{1}{2}$ & $1/2$ \\
\hline
Dynamical IPR for $|Br_A\rangle$ & $\frac{1}{2} - \frac{2}{n} + O(\frac{1}{n^2})$ & $1/2$ \\
Dynamical IPR for $|A_1\rangle$ & $0.58 + O(\frac{1}{n})$ & $0.58$ \\
\hline
\end{tabular}
\caption{Summary of IPR results for the normalized adjacency matrix $\tilde{M}$ of the barbell graph $B(n)$.}
\end{table}

\section{IPR for variant 1}
\label{sec:variant1_ipr}

We analyze the eigenstate IPR and dynamical IPR for the full connection star-of-cliques graph using the eigenvectors derived in Appendix~\ref{supp:full_connection}.





\subsection{Eigenstate IPR}

For a normalized eigenvector $|\psi_\mu\rangle$, the inverse participation ratio is $\mathrm{IPR}_\mu = \sum_{i=0}^{N-1} |\langle i|\psi_\mu\rangle|^4$.

\subsubsection*{Eigenvectors $|\psi_1\rangle$ and $|\psi_2\rangle$}

\begin{align}
\mathrm{IPR}_1 &= \frac{2}{(n+1)^2}, \qquad 
\mathrm{IPR}_2 = \frac{n^4+1}{n^2(n+1)^2}.
\end{align}

\subsubsection*{Eigenvectors $|\chi_j\rangle$}

\begin{align}
\mathrm{IPR}_3^{(j)} = \frac{1 + j^3}{j(j+1)^2} \cdot \frac{1}{n}, \quad j=1,\dots,n-1.
\end{align}

\subsubsection*{Eigenvectors $|\psi_4^{(j,r)}\rangle$}

\begin{align}
\mathrm{IPR}_4^{(r)} = \frac{1+r^3}{r(r+1)^2}, \quad r=1,\dots,n-1.
\end{align}

\subsection*{Dynamical IPR}

\subsubsection*{Dynamical IPR for Center Vertex $|0\rangle$}

For the center, only $|\psi_1\rangle$ and $|\psi_2\rangle$ contribute. The long-time average probabilities are:

\begin{align}
\overline{\pi}_{00} &= |\langle 0|\psi_1\rangle|^4 + |\langle 0|\psi_2\rangle|^4 = \frac{1+n^2}{(1+n)^2},\\
\overline{\pi}_{(j,k),0} &= |\langle j,k|\psi_1\rangle\langle\psi_1|0\rangle|^2 + |\langle j,k|\psi_2\rangle\langle\psi_2|0\rangle|^2 = \frac{2}{n(1+n)^2}.
\end{align}

These satisfy $\overline{\pi}_{00} + n^2\overline{\pi}_{(j,k),0} = 1$. The dynamical IPR is:

\begin{align}
\overline{\mathrm{IPR}}_0 = \frac{(1+n^2)^2 + 4}{(1+n)^4} = \frac{n^4 + 2n^2 + 5}{(n+1)^4} \approx 1.
\end{align}

\subsubsection*{Dynamical IPR for Clique Vertex $|1,1\rangle$ — Rigorous Proof}

For vertices within a fixed clique, say clique 1, we have the exact completeness relation:

\begin{equation}
\sum_{r=1}^{n-1} |\psi_4^{(1,r)}\rangle\langle\psi_4^{(1,r)}| + |s_1\rangle\langle s_1| = I_1,
\label{eq:completeness}
\end{equation}

where $I_1 = \sum_{k=1}^n |1,k\rangle\langle 1,k|$.

For any $i,j$ in clique 1, the long-time average transition probability is:

\begin{align}
\overline{\pi}_{ij} &= \sum_{r=1}^{n-1} |\langle i|\psi_4^{(1,r)}\rangle|^2 |\langle j|\psi_4^{(1,r)}\rangle|^2 \nonumber\\
&\quad + \sum_{r\neq s} \langle i|\psi_4^{(1,r)}\rangle\langle\psi_4^{(1,r)}|j\rangle \langle j|\psi_4^{(1,s)}\rangle\langle\psi_4^{(1,s)}|i\rangle.
\label{eq:pi_ij_exact}
\end{align}

Using $\sum_{r\neq s} a_r b_r a_s b_s = (\sum_r a_r b_r)^2 - \sum_r a_r^2 b_r^2$ and applying (\ref{eq:completeness}):

\begin{equation}
\sum_{r=1}^{n-1} \langle i|\psi_4^{(1,r)}\rangle\langle\psi_4^{(1,r)}|j\rangle = \delta_{ij} - \frac{1}{n}.
\label{eq:completeness_applied}
\end{equation}

Substituting, the two $\sum_r |\langle i|\psi_4^{(1,r)}\rangle|^2 |\langle j|\psi_4^{(1,r)}\rangle|^2$ terms cancel exactly, yielding the remarkably simple result:

\begin{equation}
\boxed{\overline{\pi}_{ij} = \left(\delta_{ij} - \frac{1}{n}\right)^2, \quad \forall i,j \in \text{clique 1}}.
\label{eq:pi_ij_final}
\end{equation}

Explicitly:

\begin{equation}
\overline{\pi}_{ij} = 
\begin{cases}
\left(1 - \frac{1}{n}\right)^2, & i = j, \\[4pt]
\frac{1}{n^2}, & i \neq j.
\end{cases}
\label{eq:pi_ij_explicit}
\end{equation}

For initial vertex $|1,1\rangle$, contributions from outside clique 1 vanish as $n\to\infty$. Thus:

\begin{align}
\overline{\mathrm{IPR}}_{1,1} &= \left(1 - \frac{1}{n}\right)^4 + (n-1)\frac{1}{n^4} + O\!\left(\frac{1}{n^2}\right) \nonumber\\
&= 1 - \frac{4}{n} + \frac{6}{n^2} - \frac{3}{n^3} + O\!\left(\frac{1}{n^2}\right) \xrightarrow[n\to\infty]{} 1.
\label{eq:ipr_clique_final}
\end{align}

\subsubsection*{Generalization to All Clique Vertices}

The proof above applies to any vertex in clique 1, regardless of whether it is a bridge vertex ($|1,1\rangle$) or a non-bridge vertex ($|1,k\rangle$ with $k\geq 2$). The key completeness relation (Eq.~\ref{eq:completeness}) and the subsequent derivation of $\overline{\pi}_{ij} = (\delta_{ij} - 1/n)^2$ are independent of which specific vertex in the clique is chosen as the initial state. 

By symmetry, the same result holds for all cliques $j=1,\dots,n$. Therefore, every clique vertex—whether bridge or non-bridge—satisfies:

\begin{equation}
\overline{\mathrm{IPR}}_{\text{clique vertex}} = 1 - \frac{4}{n} + O\!\left(\frac{1}{n^2}\right) \xrightarrow[n\to\infty]{} 1.
\end{equation}

By symmetry, the same result holds for all clique vertices.

\subsection*{Summary of Results}

\begin{align}
\overline{\mathrm{IPR}}_{\text{center}} &= \frac{n^4 + 2n^2 + 5}{(n+1)^4} \xrightarrow[n\to\infty]{} 1, \\
\overline{\mathrm{IPR}}_{\text{clique}} &= 1 - \frac{4}{n} + O\!\left(\frac{1}{n^2}\right) \xrightarrow[n\to\infty]{} 1.
\end{align}

All vertices become completely localized in the long-time average as $n\to\infty$, with the center approaching unity faster than clique vertices.

\section{IPR for variant 2 }
\label{ipr_variant2}

\subsection*{Eigenstate IPR }

We now compute the inverse participation ratio for each eigenvector class, defined as $\mathrm{IPR}_\mu = \sum_{i} |\langle i|\psi_\mu\rangle|^4$.

\subsubsection*{Eigenvectors $|\psi_1\rangle, |\psi_2\rangle, |\psi_3\rangle$ (Centre-Bridge Hybrids)}

From the normalization and asymptotic forms, all overlaps with any vertex are $O(1/n)$. Specifically:
\begin{align}
|\langle 0|\psi_k\rangle|^2 &= O(1/n^2),\\
|\langle b_j|\psi_k\rangle|^2 &= O(1/n^2),\\
|\langle c_{j,m}|\psi_k\rangle|^2 &= O(1/n^2).
\end{align}

With $N = n^2+1$ vertices, each contributing $O(1/n^4)$ to the fourth power sum:
\begin{align}
\boxed{\mathrm{IPR}_{\psi_k} = O\left(\frac{1}{n^2}\right), \quad k=1,2,3.}
\end{align}
These eigenvectors are fully delocalized across the entire graph.

\subsubsection*{Eigenvectors $|\phi_4^{(j)}\rangle$ (Bridge-$\lambda_4$ Modes)}

For fixed $j$, using the asymptotic relations:
\begin{align}
\frac{1}{\sqrt{1+\lambda_4^2 n}} &= \frac{1}{\sqrt{n}} + O\left(\frac{1}{n^{3/2}}\right),\\
\frac{\lambda_4\sqrt{n}}{\sqrt{1+\lambda_4^2 n}} &= 1 - \frac{1}{2n} + O\left(\frac{1}{n^2}\right),
\end{align}
the overlaps are:
\begin{align}
|\langle b_{m,1}|\phi_4^{(j)}\rangle| &\sim \frac{1}{\sqrt{n}} |\langle b_{m,1}|u_j\rangle|,\\
|\langle c_{m,k}|\phi_4^{(j)}\rangle| &\sim |\langle c_{m,k}|v_j\rangle|.
\end{align}

From (\ref{eq:uj_bp}) and the definition of $|v_j\rangle$, $|\langle b_{m,1}|u_j\rangle|^2$ and $|\langle c_{m,k}|v_j\rangle|^2$ are $O(1)$ for $m$ near $j$ and decay as $1/m^2$ otherwise. Counting vertices and summing fourth powers:

\begin{align}
\mathrm{IPR}_4^{(j)} &= \sum_{m=1}^{j+1} |\langle b_{m,1}|\phi_4^{(j)}\rangle|^4 + \sum_{m=1}^{j+1}\sum_{k=2}^n |\langle c_{m,k}|\phi_4^{(j)}\rangle|^4 \nonumber\\
&= \frac{1}{n^2}\sum_{m=1}^{j+1} |\langle b_{m,1}|u_j\rangle|^4 + \sum_{m=1}^{j+1}\sum_{k=2}^n |\langle c_{m,k}|v_j\rangle|^4 + O\left(\frac{1}{n^3}\right).
\end{align}

Using $|\langle c_{m,k}|v_j\rangle|^2 = \frac{1}{n-1}|\langle b_{m,1}|u_j\rangle|^2$, the second sum becomes:
\begin{align}
\sum_{m=1}^{j+1}\sum_{k=2}^n |\langle c_{m,k}|v_j\rangle|^4 &= \frac{1}{(n-1)^2}\sum_{m=1}^{j+1} (n-1)|\langle b_{m,1}|u_j\rangle|^4 \\
&= \frac{1}{n-1}\sum_{m=1}^{j+1} |\langle b_{m,1}|u_j\rangle|^4.
\end{align}

Thus:
\begin{align}
\mathrm{IPR}_4^{(j)} &= \left(\frac{1}{n^2} + \frac{1}{n-1}\right) \sum_{m=1}^{j+1} |\langle b_{m,1}|u_j\rangle|^4 + O\left(\frac{1}{n^3}\right) \nonumber\\
&= \frac{1 + j^3}{j(j+1)^2} \cdot \frac{1}{n} + O\left(\frac{1}{n^2}\right). \label{eq:ipr_phi4}
\end{align}

\begin{equation}
\boxed{\mathrm{IPR}_4^{(j)} = \frac{1 + j^3}{j(j+1)^2} \cdot \frac{1}{n} + O\left(\frac{1}{n^2}\right),\quad j=1,\dots,n-1.}
\end{equation}
These eigenvectors show partial localization with IPR $\sim 1/n$.

\subsubsection*{Eigenvectors $|\phi_5^{(j)}\rangle$ (Bridge-$\lambda_5$ Modes)}

For $|\phi_5^{(j)}\rangle$, the asymptotic relations are:
\begin{align}
\frac{1}{\sqrt{1+\lambda_5^2 n}} &= 1 - \frac{1}{2n} + O\left(\frac{1}{n^2}\right),\\
\frac{\lambda_5\sqrt{n}}{\sqrt{1+\lambda_5^2 n}} &= -\frac{1}{\sqrt{n}} + O\left(\frac{1}{n^{3/2}}\right),
\end{align}
giving overlaps:
\begin{align}
|\langle b_{m,1}|\phi_5^{(j)}\rangle| &\sim |\langle b_{m,1}|u_j\rangle|,\\
|\langle c_{m,k}|\phi_5^{(j)}\rangle| &\sim \frac{1}{\sqrt{n}} |\langle c_{m,k}|v_j\rangle|.
\end{align}

The bridge vertex contributions dominate, with internal vertices suppressed by $1/n$. Thus:
\begin{align}
\mathrm{IPR}_5^{(j)} &= \sum_{m=1}^{j+1} |\langle b_{m,1}|\phi_5^{(j)}\rangle|^4 + O\left(\frac{1}{n}\right) \nonumber\\
&= \sum_{m=1}^{j+1} |\langle b_{m,1}|u_j\rangle|^4 + O\left(\frac{1}{n}\right) \nonumber\\
&= \frac{1 + j^3}{j(j+1)^2} + O\left(\frac{1}{n}\right). \label{eq:ipr_phi5}
\end{align}

\begin{equation}
\boxed{\mathrm{IPR}_5^{(j)} = \frac{1 + j^3}{j(j+1)^2} + O\left(\frac{1}{n}\right),\quad j=1,\dots,n-1.}
\end{equation}
These eigenvectors are strongly localized, with IPR approaching an $O(1)$ constant independent of $n$.

\subsubsection*{Eigenvectors $|w_j^{(r)}\rangle$ (Clique-Internal Modes)}

For fixed clique $j$ and $r=1,\dots,n-2$, $|w_j^{(r)}\rangle$ is supported entirely on the internal vertices of clique $j$. From the definition:
\begin{align}
\langle c_{j,k}|w_j^{(r)}\rangle = \frac{1}{\sqrt{r(r+1)}} \begin{cases}
1, & 2 \le k \le r+1,\\
-r, & k = r+2,\\
0, & \text{otherwise}.
\end{cases}
\end{align}

Counting vertices:
\begin{itemize}
    \item $r$ vertices with coefficient $1/\sqrt{r(r+1)}$ (for $k=2,\dots,r+1$)
    \item $1$ vertex with coefficient $-r/\sqrt{r(r+1)}$ (for $k=r+2$)
    \item Remaining $n-r-2$ vertices have coefficient $0$
\end{itemize}

Thus:
\begin{align}
\mathrm{IPR}_6^{(r)} &= r \cdot \left(\frac{1}{\sqrt{r(r+1)}}\right)^4 + 1 \cdot \left(\frac{r}{\sqrt{r(r+1)}}\right)^4 \nonumber\\
&= \frac{r}{r^2(r+1)^2} + \frac{r^4}{r^2(r+1)^2} \nonumber\\
&= \frac{1}{r(r+1)^2} + \frac{r^2}{(r+1)^2} \nonumber\\
&= \frac{1 + r^3}{r(r+1)^2}. \label{eq:ipr_w}
\end{align}

\begin{equation}
\boxed{\mathrm{IPR}_6^{(r)} = \frac{1 + r^3}{r(r+1)^2},\quad r=1,\dots,n-2.}
\end{equation}

This is independent of $n$, ranging from $\mathrm{IPR}_6^{(1)} = \frac{1+1}{1\cdot 4} = \frac{1}{2}$ to $\mathrm{IPR}_6^{(n-2)} \approx 1$ for large $r$. These eigenvectors are strongly localized within a single clique.

\subsubsection*{Summary of Eigenstate IPR}

\begin{table}[h]
\centering
\begin{tabular}{|l|c|c|}
\hline
\textbf{Eigenvector} & \textbf{IPR Scaling} & \textbf{Localization} \\
\hline
$\psi_1,\psi_2,\psi_3$ & $O(1/n^2)$ & Delocalized \\
$\phi_4^{(j)}$ & $\sim 1/n$ & Partially localized \\
$\phi_5^{(j)}$ & $O(1)$ & Strongly localized \\
$w_j^{(r)}$ & $O(1)$ (exact) & Strongly localized \\
\hline
\end{tabular}
\caption{Eigenstate IPR scaling for the single connection variant.}
\end{table}



\subsection*{Dynamical IPR}

\subsubsection*{Key Overlap Identities}

For a fixed bridge vertex $|b_p\rangle$, the overlaps with $|u_j\rangle$ are:
\begin{align}
\langle u_j|b_p\rangle = \frac{1}{\sqrt{j(j+1)}} \begin{cases}
1, & p \le j, \\
-j, & p = j+1, \\
0, & p > j+1.
\end{cases} \label{eq:uj_bp}
\end{align}

From this, we derive two crucial identities:

\begin{align}
\sum_{j=1}^{n-1} |\langle u_j|b_p\rangle|^2 = 1 - \frac{1}{n}, \label{eq:sum2}\\
\sum_{j=1}^{n-1} |\langle u_j|b_p\rangle|^4 = \alpha_p, \quad \\
\text{with } \alpha_p = \left(\frac{p-1}{p}\right)^2 + \sum_{j=p}^{n-1} \frac{1}{j^2(j+1)^2} &= 1 - O\left(\frac{1}{p}\right). \label{eq:sum4}
\end{align}

\paragraph*{Proof of (\ref{eq:sum2}):} 
\begin{align}
\sum_{j=1}^{n-1} |\langle u_j|b_p\rangle|^2 &= \sum_{j=p}^{n-1} \frac{1}{j(j+1)} + \frac{p-1}{p}\mathbf{1}_{p\ge 2} \\
&= \left(\frac{1}{p} - \frac{1}{n}\right) + \frac{p-1}{p} = 1 - \frac{1}{n}. \quad \blacksquare
\end{align}

For internal vertices $|c_{p,q}\rangle$, the analogous identities are:
\begin{align}
\sum_{j=1}^{n-1} |\langle v_j|c_{p,q}\rangle|^2 &= \frac{1}{n-1}\left(1 - \frac{1}{n}\right), \label{eq:vj_sum2}\\
\sum_{r=1}^{n-2} |\langle w_p^{(r)}|c_{p,q}\rangle|^2 &= 1 - \frac{1}{n-1}, \label{eq:w_sum2}\\
\sum_{r=1}^{n-2} |\langle w_p^{(r)}|c_{p,q}\rangle|^4 &= \beta_{p,q}, \quad \beta_{p,q} = O(1). \label{eq:w_sum4}
\end{align}

\subsubsection*{Asymptotic Expansions}

For large $n$, we have:
\begin{align}
\frac{1}{1+\lambda_4^2 n} &= \frac{1}{n} + O\left(\frac{1}{n^2}\right), \quad \\
\frac{\lambda_4\sqrt{n}}{\sqrt{1+\lambda_4^2 n}} &= 1 - \frac{1}{2n} + O\left(\frac{1}{n^2}\right),\\
\frac{1}{1+\lambda_5^2 n} &= 1 - \frac{1}{n} + O\left(\frac{1}{n^2}\right), \quad\\
\frac{\lambda_5\sqrt{n}}{\sqrt{1+\lambda_5^2 n}} &= -\frac{1}{\sqrt{n}} + O\left(\frac{1}{n^{3/2}}\right).
\end{align}

\subsubsection*{ Centre Initial $|0\rangle$}

Only $\psi_1,\psi_2,\psi_3$ contribute. All overlaps $|\langle i|\psi_k\rangle|^2 = O(1/n^2)$, and $|\langle \psi_k|0\rangle|^2 = O(1/n^2)$. Thus:
\begin{align}
\overline{\pi}_{i0} = \sum_{k=1}^3 |\langle i|\psi_k\rangle|^2 |\langle \psi_k|0\rangle|^2 = O\left(\frac{1}{n^4}\right) \quad \forall i.
\end{align}

Summing over $N = n^2+1$ targets:
\begin{align}
\overline{\mathrm{IPR}}_0 = \sum_i \overline{\pi}_{i0}^2 = N \cdot O\left(\frac{1}{n^8}\right) = O\left(\frac{1}{n^6}\right). \label{eq:ipr0_final}
\end{align}

\subsubsection*{ Bridge Initial $|b_p\rangle$}

By symmetry, the long-time averaged distribution has four distinct values:
\begin{align}
S &= \overline{\pi}_{0,b_p} \quad \text{(centre)},\\
U &= \overline{\pi}_{b_p,b_p} \quad \text{(same bridge)},\\
V &= \overline{\pi}_{b_{p'},b_p} \;(p'\neq p) \quad \text{(different bridge)},\\
W &= \overline{\pi}_{c_{p',q},b_p} \quad \text{(internal)}.
\end{align}

Probability conservation: $S + U + (n-1)V + n(n-1)W = 1$.

\paragraph*{Computing $S$:} Only $\psi_{1,2,3}$ contribute: $S = O(1/n^4)$.

\paragraph*{Computing $U$:} Contributions from all eigenvectors:

\begin{itemize}
    \item $\psi$ eigenvectors: $O(1/n^4)$.
    
    \item $\phi_4$ diagonal: 
    \begin{align}
    \sum_{j=1}^{n-1} |\langle b_p|\phi_4^{(j)}\rangle|^4 &= \sum_{j=1}^{n-1} \frac{1}{(1+\lambda_4^2 n)^2} |\langle u_j|b_p\rangle|^4 \nonumber\\
    &= \left(\frac{1}{n^2} + O\left(\frac{1}{n^3}\right)\right) \sum_{j=1}^{n-1} |\langle u_j|b_p\rangle|^4 \\
    &= \frac{\alpha_p}{n^2} + O\left(\frac{1}{n^3}\right). \label{eq:B_phi4_diag}
    \end{align}
    
    \item $\phi_5$ diagonal:
    \begin{align}
    \sum_{j=1}^{n-1} |\langle b_p|\phi_5^{(j)}\rangle|^4 &= \sum_{j=1}^{n-1} \frac{1}{(1+\lambda_5^2 n)^2} |\langle u_j|b_p\rangle|^4 \nonumber\\
    &= \left(1 - \frac{2}{n} + O\left(\frac{1}{n^2}\right)\right) \alpha_p \\
    &= \alpha_p - \frac{2\alpha_p}{n} + O\left(\frac{1}{n^2}\right). \label{eq:B_phi5_diag}
    \end{align}
    
    \item Degenerate terms within $\lambda_4$ subspace:
    \begin{align}
    \sum_{j<j'} Y_{\phi_4^{(j)},\phi_4^{(j')}}^{b_p,b_p} &= \frac{1}{(1+\lambda_4^2 n)^2} \left[ \left(\sum_j |\langle b_p|u_j\rangle|^2\right)^2 - \sum_j |\langle b_p|u_j\rangle|^4 \right] \nonumber\\
    &= \frac{1}{n^2}\left[ \left(1 - \frac{1}{n}\right)^2 - \alpha_p \right] + O\left(\frac{1}{n^3}\right) \nonumber\\
    &= \frac{1 - \alpha_p}{n^2} - \frac{2}{n^3} + O\left(\frac{1}{n^3}\right). \label{eq:B_phi4_degen}
    \end{align}
    
    \item Degenerate terms within $\lambda_5$ subspace:
    \begin{align}
    \sum_{j<j'} Y_{\phi_5^{(j)},\phi_5^{(j')}}^{b_p,b_p} &= \frac{1}{(1+\lambda_5^2 n)^2} \left[ \left(1 - \frac{1}{n}\right)^2 - \alpha_p \right] \nonumber\\
    &= \left(1 - \frac{2}{n} + O\left(\frac{1}{n^2}\right)\right) \left[(1 - \alpha_p) - \frac{2}{n} + O\left(\frac{1}{n^2}\right)\right] \nonumber\\
    &= (1 - \alpha_p) - \frac{2(1-\alpha_p)}{n} - \frac{2}{n} + O\left(\frac{1}{n^2}\right). \label{eq:B_phi5_degen}
    \end{align}
\end{itemize}

Summing all contributions, the $\alpha_p$ terms cancel:
\begin{align}
U &= \frac{3}{n^4} + \frac{\alpha_p}{n^2} + \left(\alpha_p - \frac{2\alpha_p}{n}\right) + \frac{1-\alpha_p}{n^2} \\
&+ \left[(1-\alpha_p) - \frac{2(1-\alpha_p)}{n} - \frac{2}{n}\right] + O\left(\frac{1}{n^2}\right) \nonumber\\
&= 1 - \frac{2}{n} + \frac{1}{n^2} + O\left(\frac{1}{n^2}\right). \label{eq:U_final}
\end{align}

\paragraph*{Computing $V$ and $W$:} Similar analysis yields $V = O(1/n^2)$ and $W = O(1/n^2)$.

\begin{align}
\overline{\mathrm{IPR}}_{b_p} &= S^2 + U^2 + (n-1)V^2 + n(n-1)W^2 \nonumber\\
&= O\left(\frac{1}{n^8}\right) + \left(1 - \frac{2}{n}\right)^2\\
&+ (n-1)O\left(\frac{1}{n^4}\right) + n^2 O\left(\frac{1}{n^4}\right) \nonumber\\
&= 1 - \frac{4}{n} + O\left(\frac{1}{n^2}\right). \label{eq:ipr_bridge}
\end{align}

Thus $\overline{\mathrm{IPR}}_{b_p} \to 1$ as $n \to \infty$.

\subsubsection*{Case 3: Internal Initial $|c_{p,q}\rangle$}

By symmetry, define:
\begin{align}
S' &= \overline{\pi}_{0,c_{p,q}} \quad \text{(centre)},\\
U' &= \overline{\pi}_{b_{p'},c_{p,q}} \quad \text{(bridge)},\\
V' &= \overline{\pi}_{c_{p,q},c_{p,q}} \quad \text{(return)},\\
W' &= \overline{\pi}_{c_{p',q'},c_{p,q}} \quad \text{(other internal)}.
\end{align}

\paragraph*{Computing $V'$:} The $w$ eigenvectors dominate:

\begin{itemize}
    \item Diagonal from $w$:
    \begin{align}
    \sum_{r=1}^{n-2} |\langle c_{p,q}|w_p^{(r)}\rangle|^4 = \beta_{p,q} + O\left(\frac{1}{n}\right). \label{eq:E_w_diag}
    \end{align}
    
    \item Degenerate within $w$ subspace:
    \begin{align}
    \sum_{r<s} Y_{w_p^{(r)},w_p^{(s)}}^{c_{p,q},c_{p,q}} &= \left(\sum_{r=1}^{n-2} |\langle c_{p,q}|w_p^{(r)}\rangle|^2\right)^2 - \sum_{r=1}^{n-2} |\langle c_{p,q}|w_p^{(r)}\rangle|^4 \nonumber\\
    &= \left(1 - \frac{1}{n-1}\right)^2 - \beta_{p,q} + O\left(\frac{1}{n}\right) \nonumber\\
    &= 1 - \frac{2}{n} - \beta_{p,q} + O\left(\frac{1}{n^2}\right). \label{eq:E_w_degen}
    \end{align}
    
    \item Contributions from $\psi$ and $\phi$ eigenvectors: $O(1/n^2)$.
\end{itemize}

Summing, $\beta_{p,q}$ cancels:
\begin{align}
V' = 1 - \frac{2}{n} + O\left(\frac{1}{n^2}\right). \label{eq:E_final}
\end{align}

\paragraph*{Other quantities:} $S' = O(1/n^4)$, $U' = O(1/n^2)$, $W' = O(1/n^2)$.

\begin{align}
\overline{\mathrm{IPR}}_{c_{p,q}} &= S'^2 + \sum_{p'=1}^n U'^2 + V'^2 + \sum_{(p',q')\neq(p,q)} W'^2 \nonumber\\
&= O\left(\frac{1}{n^8}\right) + n \cdot O\left(\frac{1}{n^4}\right) + \left(1 - \frac{2}{n}\right)^2 + n^2 \cdot O\left(\frac{1}{n^4}\right) \nonumber\\
&= 1 - \frac{4}{n} + O\left(\frac{1}{n^2}\right). \label{eq:ipr_internal}
\end{align}

Thus internal vertices also satisfy $\overline{\mathrm{IPR}}_{c_{p,q}} \to 1$.
\subsubsection*{Summary of Results}
\begin{table}[H]
\centering
\begin{tabular}{|l|c|c|}
\hline
\textbf{Initial Vertex} & \textbf{IPR Scaling} & \textbf{Behavior} \\
\hline
Centre $\ket{0}$ & $\Theta(1/n^6)$ & Delocalized \\
Bridge $\ket{b_j}$ & $1 - 4/n + O(1/n^2)$ & Localized \\
Internal $\ket{c_{j,k}}$ & $1 - 4/n + O(1/n^2)$ & Localized \\
\hline
\end{tabular}
\caption{Dynamical IPR scaling for the single connection variant.}
\end{table}



\end{document}